\documentclass[11pt]{article}
\usepackage[usenames,dvipsnames,svgnames,table]{xcolor}
\usepackage{enumitem}
\usepackage{graphicx}
\usepackage{fancybox}
\usepackage{comment}
\usepackage{xcolor}
\usepackage{nameref}

\definecolor{ForestGreen}{rgb}{0.1333,0.5451,0.1333}
\definecolor{DarkRed}{rgb}{0.8,0,0}
\definecolor{Red}{rgb}{1,0,0}
\usepackage[linktocpage=true,
pagebackref=true,colorlinks,
linkcolor=ForestGreen,citecolor=ForestGreen,
bookmarks,bookmarksopen,bookmarksnumbered]
{hyperref}

\usepackage{amsmath}
\usepackage{amssymb}
\usepackage{amsthm}

\usepackage[ruled, noend, linesnumbered]{algorithm2e}

\usepackage{subfig}

\usepackage{thm-restate}

\usepackage{float}
\usepackage{cleveref}

\newtheorem{theorem}{Theorem}[section]

\newtheorem{lemma}[theorem]{Lemma}
\newtheorem{observation}[theorem]{Observation}

\newtheorem{claim}[theorem]{Claim}

\newtheorem{fact}[theorem]{Fact}
\newtheorem{invariant}[theorem]{Invariant}

\newtheorem{question}[theorem]{Question}

\newtheorem{definition}[theorem]{Definition}

\newtheorem*{theorem*}{Theorem}
\newtheorem*{corollary*}{Corollary}
\newtheorem*{conjecture*}{Conjecture}
\newtheorem*{lemma*}{Lemma}
\newtheorem*{thm*}{Theorem}
\newtheorem*{prop*}{Proposition}
\newtheorem*{obs*}{Observation}
\newtheorem*{definition*}{Definition}

\newtheorem*{remark*}{Remark}
\newtheorem*{rec*}{Recommendation}

\newenvironment{fminipage}%
  {\begin{Sbox}\begin{minipage}}%
  {\end{minipage}\end{Sbox}\fbox{\TheSbox}}

\def\defeq{\stackrel{\mathrm{def}}{=}}

\def\ceil#1{\left\lceil #1 \right\rceil}

\def\norm#1{\left\| #1 \right\|}

\newcommand\ddelta{\boldsymbol{\delta}}

\newcommand\cc{\boldsymbol{\mathit{c}}}
\newcommand\dd{\boldsymbol{\mathit{d}}}

\newcommand\ff{\boldsymbol{\mathit{f}}}

\newcommand\qq{\boldsymbol{\mathit{q}}}

\renewcommand\ss{\boldsymbol{\mathit{s}}}
\def\tt{\boldsymbol{\mathit{t}}}

\newcommand\xx{\boldsymbol{\mathit{x}}}

\newcommand\veczero{\boldsymbol{0}}
\newcommand\vecone{\boldsymbol{1}}

\renewcommand{\deg}{\mathbf{deg}}

\newcommand\BB{\boldsymbol{\mathit{B}}}

\newcommand\poly{\mathrm{poly}}

\newcommand\R{\mathbb{R}}

\usepackage{fullpage}

\newcommand{\thatchaphol}[1]{{\bf \color{DarkRed} Thatchaphol: #1}}
\def\tfootnote#1{\marginpar{$\leftarrow$\fbox{TS}}\footnote{$\Rightarrow$~{\sf\textcolor{purple}{#1 --Thatchaphol}}}}
\newcommand{\mprobst}[1]{{\bf \color{Red} Max: #1}}
\newcommand{\simon}[1]{{\bf \color{olive} Simon: #1}}
\renewcommand{\tfootnote}[1]{{\bf \color{olive} Thatchaphol: #1}}

\renewcommand{\thatchaphol}[1]{}
\renewcommand{\tfootnote}[1]{}
\renewcommand{\mprobst}[1]{}
\renewcommand{\simon}[1]{}

\DeclareMathOperator{\polylog}{polylog}

\DeclareMathOperator{\vol}{\mathbf{vol}}
\DeclareMathOperator{\dist}{dist}

\DeclareMathOperator{\supp}{supp}

\newcommand\dir{\overrightarrow}
\newcommand*\samethanks[1][\value{footnote}]{\footnotemark[#1]}

\title{Expander Pruning with Polylogarithmic Worst-Case Recourse and Update Time}
\author{Simon Meierhans\thanks{The research leading to these results has received funding from grant no. 200021 204787 of the Swiss National Science Foundation. Simon Meierhans is supported by a Google PhD Fellowship.} \\
ETH Zurich \\
mesimon@inf.ethz.ch 
\and  
Maximilian Probst Gutenberg\samethanks[1] \\
ETH Zurich \\
maximilian.probst@inf.ethz.ch
\and
Thatchaphol Saranurak\thanks{
       Supported by NSF Grant CCF-2238138. 
        Part of this work was done while at INSAIT, Sofia University ``St. Kliment Ohridski'', Bulgaria. This work was partially funded from the Ministry of Education and Science of Bulgaria (support for INSAIT, part of the Bulgarian National Roadmap for Research Infrastructure).}\\
University of Michigan\\
thsa@umich.edu
}
\begin{document}

\maketitle

\begin{abstract}
Expander graphs are known to be robust to edge deletions in the following sense: for any online sequence of edge deletions $e_1, e_2, \ldots, e_k$ to an $m$-edge graph $G$ that is initially a $\phi$-expander, the algorithm can grow a set $P \subseteq V$ such that at any time $t$, $G[V \setminus P]$ is an expander of the same quality as the initial graph $G$ up to a constant factor and the set $P$ has volume at most $O(t/\phi)$. However, currently, there is no algorithm to grow $P$ with low \emph{worst-case} recourse that achieves any non-trivial guarantee.

In this work, we present an algorithm that achieves near-optimal guarantees: we give an algorithm that grows $P$ only by $\tilde{O}(1/\phi^2)$ vertices per time step and ensures that $G[V \setminus P]$ remains $\tilde{\Omega}(\phi)$-expander at any time.\footnote{In this article, $\tilde{O}(\cdot)$ and $\tilde{\Omega}(\cdot)$ hide poly-logarithmic factors in the number of edges $m$.} Even more excitingly, our algorithm is extremely efficient: it can process each update in near-optimal \emph{worst-case} update time $\tilde{O}(1/\phi^2)$. This affirmatively answers the main open question posed in \cite{saranurak2021expander} whether such an algorithm exists.

By combining our results with recent techniques in \cite{bernstein2022fully}, we obtain the first adaptive algorithms to maintain spanners, cut and spectral sparsifiers with $\tilde{O}(n)$ edges and polylogarithmic approximation guarantees, \emph{worst-case} update time and recourse. More generally, we believe that worst-case pruning is an essential tool for obtaining worst-case guarantees in dynamic graph algorithms and online algorithms.
\end{abstract}

\newpage

\tableofcontents
\pagebreak

\section{Introduction}

The notion of an expander formalizes and quantifies what makes a graph 'well-connected'. Intuitively, a graph is 'well-connected' if every cut $(S, V \setminus S)$ sends a large fraction of the edges incident to the smaller side over the cut. Formally, we say $(S, V \setminus S)$ is $\phi$-\emph{sparse} if $\frac{|E_G(S, V \setminus S)|}{\min\{\vol_G(S), \vol_G(V \setminus S)\}} < \phi$, and a $\phi$-expander is a graph that contains no $\phi$-sparse cut $(S, V \setminus S)$. 

One of the many desirable properties of expander graphs is that they are robust to edge updates, particularly edge deletions to the graph \cite{bagchi2004effect,saranurak2021expander}. Say a set $B \subset E$ of edges is deleted from a $\phi$-expander $G$. It is well-known that a vertex set $A \subseteq V$ with volume scaling linearly in $B$ (and $1/\phi$), $\vol_G(A) \leq O(|B|/\phi)$ exists such that after \emph{pruning} $A$ from $G$, we recover that the graph $(G\setminus B)[V \setminus A]$ is a  $\Omega(\phi)$-expander. The proof behind this theorem can be extended to yield the following result: given a sequence of edges $e_1, e_2, \ldots, e_k$ arriving one by one, there is an online algorithm that grows $A$ monotonically over time such that at any time $t$, $A$ has volume at most $O(t/\phi)$ and $G \setminus \{e_1, e_2, \ldots, e_t\}[V \setminus A]$ is a $\Omega(\phi)$-expander.

However, to the best of our knowledge, there is currently no algorithm growing the set $A$ little by little, i.e. that achieves a small worst-case recourse for $A$. In fact,  no non-trivial guarantees are currently achieved by any known algorithm. We thus ask the following question.

\begin{question}\label{question1:worstCaseRec}
Given a $\phi$-expander graph $G=(V,E)$ undergoing a sequence of edge deletions $e_1, e_2, \ldots, e_k$, is there an online algorithm that maintains a set $A$, with guarantees as above, with \emph{worst-case recourse} $\tilde{O}(1/\poly(\phi))$? 
\end{question}

A seemingly much harder question is whether there is an algorithm to maintain the set $A$ for an online sequence of edge deletions that spends only little time processing each edge deletion before publishing the updated set $A$. 

\begin{question}\label{question2:worstCase}
Given a $\phi$-expander graph $G=(V,E)$ undergoing a sequence of edge deletions $e_1, e_2, \ldots, e_k$, is there an algorithm that maintains the set $A$, with guarantees as above, that requires at most $\tilde{O}(1/\poly(\phi))$ time per update?
\end{question}

Surprisingly, while no algorithm is known that achieves the bound asked for in \Cref{question2:worstCase}, the algorithms in \cite{NSW17, saranurak2021expander, bernstein2022fully, hua2023maintaining, sulser2024} achieve worst-case update time $n^{o(1)}/ \poly(\phi)$, which already comes fairly close. 

At this point, the astute reader might wonder how such a runtime bound can be achieved while no non-trivial worst-case recourse bound is achieved simultaneously. The answer is that the algorithms mentioned above cannot update the set $A$ explicitly, but instead, they build a new set $A'$ in the background and then switch the pointer from $A$ to $A'$. The difference between these two sets $A$ and $A'$ might be as large as the set $A$ itself. This however makes it cumbersome (if not impossible) to use these expander pruning algorithms as subroutines in graph algorithms with worst-case time guarantees, one of their main area of application.

\paragraph{Our Contribution.} In this article, we resolve both questions and obtain an algorithm with low worst-case recourse and update time.

\begin{theorem}
    \label{thm:main_thm}
    Given an $m$-edge $\phi$-expander graph $G = (V, E)$ and a sequence of up to $\tilde{\Omega}(\phi \cdot  m)$ edge deletions to $G$. There is a deterministic algorithm that processes each edge deletion in time $\tilde{O}(1/\phi^2)$ and adds at most $\tilde{O}(1/\phi^2)$ vertices to the initially empty set $A$. Further, at any time, $G[V \setminus A]$ is a $\tilde{\Omega}(\phi)$-expander.    
\end{theorem}

Note that our algorithm is deterministic and retains the quality of the expander pruning maintenance up to a poly-logarithmic factor in $m$. We further believe that the new techniques developed to obtain \Cref{thm:main_thm} are of general interest, significantly advance our understanding of expander maintenance techniques, and have enormous potential to help solve various other open questions relating to expander pruning. We present these ideas in detail in \Cref{sec:overview}. We refer the reader to \Cref{thm:double_low_recourse} for a more detailed statement of \Cref{thm:main_thm}.

\paragraph{Applications in Dynamic Graph Algorithms.} 
Our result from \Cref{thm:main_thm} is highly motivated since expander pruning is a fundamental tool in dynamic graph algorithms \cite{ns17, NSW17, saranurak2021expander,chuzhoy2019new, chuzhoy2020deterministic, bernstein2020deterministic, chuzhoy2021deterministic, chuzhoy2021decremental, goranci2021expander, jin2022fully, bernstein2022deterministic, kyng2023dynamic, goranci2023fully, haeupler2024dynamic, jin2024fully, chuzhoy2024faster, el2025fully}, and worst-case update time guarantees are heavily sought after in the area. 

In particular, combining our result with the techniques from \cite{bernstein2022fully}, we immediately obtain the first algorithms to maintain spanners, cut and spectral sparsifiers with poly-logarithmic worst-case update time and recourse.

\begin{theorem}\label{thm:spannerThm}
Given an $m$-edge $n$-vertex weighted graph $G=(V,E,w)$ with aspect ratio $W$ undergoing at most $\tilde{O}(m)$ edge updates in the form of insertions and deletions. Then, there is a randomized algorithm, that explicitly maintains an $\tilde{O}(n \log mW)$-edge graph $H$ undergoing at most $\tilde{O}(\log mW)$ edge updates per adversarial edge update to $G$ such that at any time
\begin{itemize}
    \item (Distance-Preserving) for any two vertices $u,v \in V$, $\dist_H(u,v) \in \tilde{\Theta}(\dist_G(u,v))$, and
    \item (Cut-Preserving) for any cut $(S, V \setminus S)$, we have $w_H(E_H(S, V \setminus S)) \in \tilde{\Theta}(w_G(E_G(S, V \setminus S)))$, and
    \item (Spectral Sparsifier) for any vector ${x} \in \mathbb{R}^n$, we have that ${x}^T {L}_H {x} \in \tilde{\Theta}({x}^T {L}_G {x})$ where $L_H$ and $L_G$ are the graph Laplacians of $H$ and $G$ respectively.
\end{itemize}
The algorithm takes $\tilde{O}(m \log mW)$ time on initialization and then requires worst-case randomized update time $\tilde{O}(\log mW)$, works against an adaptive adversary\footnote{In the adaptive adversary model, the updates to $G$ can be designed by the algorithm on-the-go and thus with consideration of the output of the dynamic graph algorithm. This is in stark contrast to the \emph{oblivious} adversary model where the update sequence is required to be fixed by the adversary before the dynamic graph algorithm is initialized.}, and succeeds with high probability. 
\end{theorem}

Previously, this result was only known when allowing for either \emph{amortization}, an oblivious adversary assumption, or $m^{o(1)}$ worst-case update time (see \cite{bernstein2022fully}).

Since spanners and cut/spectral sparsifiers are important primitives in graph algorithms, we hope that our improvement paves the way for algorithms with polylogarithmic worst-case update times for many dynamic shortest paths and cut/flow problems. We point the reader to the following exciting results  \cite{ns17, NSW17, saranurak2021expander,chuzhoy2020deterministic, goranci2021expander, kyng2023dynamic, haeupler2024dynamic} that already achieve subpolynomial worst-case times and build on (dynamic) expander techniques.

\paragraph{Fully-Dynamic Connectivity with Worst-Case Time via Expander Pruning.} Our result is further motivated by the recent approach to the fully-dynamic connectivity problem \cite{ns17, NSW17, saranurak2021expander,chuzhoy2020deterministic}. The fully dynamic connectivity problem asks whether we can process edge updates efficiently and supports queries for whether any two given vertices $u,v \in V$ are connected in the current graph. Besides its fundamental nature and practical importance, the fully-dynamic connectivity problem has proven to be a test bed for developing a general toolbox for dynamic graph algorithms. 

While already studied in the early 1980s \cite{harel1982line, frederickson1983data}, the first randomized algorithm with polylogarithmic update and query time was first given in \cite{henzinger1995randomized} and this algorithm was later de-randomized \cite{holm2001poly}. But both of these algorithms only achieved polylogarithmic \emph{amortized} update and worst-case query times. In \cite{kapron2013dynamic}, the first randomized algorithm was given to obtain polylogarithmic \emph{worst-case} update and query time, a result that obtained the SODA best paper award in 2013. 

Since, the main open question in this line of research is whether a deterministic algorithm with worst-case update time exists. This question has been partially resolved by a sequence of impressive results \cite{ns17, wul17, NSW17, chuzhoy2020deterministic} where a deterministic algorithm with $m^{o(1)}$ \emph{worst-case} update and query time was given. 

This recent approach reduces the fully dynamic connectivity problem to pruning a decremental expander graph. In our work, we give the first algorithm that achieves \emph{polylogarithmic worst-case update time} for pruning decremental expanders and our algorithm is deterministic. Thus, our result can be interpreted as removing a fundamental obstacle on the path toward the holy grail in the area: a deterministic polylogarithmic worst-case update time algorithm for fully-dynamic connectivity. We point out, however, that in the reduction to decremental expander pruning, the framework in \cite{ns17, wul17, NSW17, chuzhoy2020deterministic} already loses an $m^{o(1)}$ factor as they require approximately finding sparsest cuts deterministically, and since their hierarchy incurs logarithmic recourse over each level.

\paragraph{Applications in Online Algorithms.} Naturally, the low worst-case recourse achieved by our algorithm in \Cref{thm:main_thm} also seems to be a powerful tool for online algorithms. 

In fact, the algorithms in \Cref{thm:spannerThm} are the first to achieve polylogarithmic recourse for spanners, cut- and spectral sparsifiers in the online setting against an adaptive adversary.

As another example, our algorithm significantly simplifies the construction in \cite{gupta2022online} to maintain an orientation of a dynamic graph of polylogarithmic discrepancy with polylogarithmic \emph{amortized} recourse. Here, an orientation's discrepancy is defined as the maximum difference between in- and out-degree achieved by any vertex with respect to the chosen orientation. To obtain similar \emph{worst-case} recourse bounds, a more careful combination of our tools with their framework is necessary, possibly requiring new ideas. We believe that such new ideas might then also lead to de-amortization of the recourse bounds in \cite{gupta2020online} which considers the online car-pooling problem.

\section{Overview}
\label{sec:overview}

In this overview, we present the techniques to achieve \Cref{thm:main_thm}. To illustrate our ideas, we focus on obtaining an online algorithm that achieves low worst-case recourse. We then outline how to implement these ideas efficiently to obtain an algorithm that also achieves polylogarithmic worst-case update time.

\paragraph{Review: Expander Pruning with Low Amortized Recourse.} Before we present our techniques, we first review how to maintain an expander with low amortized recourse. Here, we focus on a recent proof given most explicitly in \cite{saranurak2021expander} that relies on the max-flow min-cut theorem and that has also been heavily exploited algorithmically. 

We start by describing expander pruning for the special case that was previously mentioned in the introduction: Given a $\phi$-expander $G$ and a set of edges $B \subseteq E$, there is a set $A$ of small volume $O(|B|/\phi)$ such that after \emph{pruning} $A$ from $G$, we recover that the graph $(G\setminus B)[V \setminus A]$ is a $\Omega(\phi)$-expander. 

For this proof, let us define the following flow problem.

\begin{definition}[\cite{saranurak2021expander} flow instance]
    \label{def:overview_flow_instance}
    Given a $\phi$-expander graph $G = (V, E)$ and a set $B \subset E$ of edge deletions. Consider the flow network on $G \setminus B$ where for each edge $e \in B$, $8/\phi$ source flow is added to both endpoints of $e$, every vertex $v \in V$ is a sink of capacity $\deg(v)$, and all capacities are of value $8/\phi$. 
\end{definition}

Let $\ff$ be a maximum flow for the instance above, i.e. a flow that routes the maximum amount of flow from sources to sinks. Consider the most-balanced cut $(A, V \setminus A)$ such that every edge in the cut routes $8/\phi$ units of flow from $A$ to $V \setminus A$. By \emph{most-balanced}, we mean the cut that maximizes $\min\{\vol(A), \vol(V \setminus A)\}$. Let us proof that $A$ satisfies the properties claimed above:

\begin{itemize}
    \item \underline{$\vol_G(A) = O(|B|/\phi)$:} From the max-flow min-cut theorem, we have that $\ff$ is such that all sinks in $A$ absorb the maximum amount possible, i.e. each vertex $a \in A$ absorbs $\deg_G(a)$ units of flow. Thus, the flow $\ff$ routes at least $\vol_G(A) = \sum_{a \in A} \deg_G(a)$ flow units from sources to sinks. But since each edge $e \in B$ contributes at most $16/\phi$ units to the sources ($8/\phi$ per endpoint), this yields an upper bound on the volume of $A$ of $16|B|/\phi$. 

    \item \underline{$(G \setminus B)[V \setminus A]$ is $\Omega(\phi)$-expander:} Consider any cut  $(S, \bar{S} = V \setminus (S \cup A))$ in the graph $G[V \setminus A]$. Initially, $E_G(S, V \setminus S)$ contains at least $\phi \cdot \vol_G(S)$ edges.\footnote{Note that this claim is slightly incorrect as the volumes in $G$ and $(G\setminus B)[V \setminus A]$ are not proportional and thus the smaller side of a cut does not necessarily align. However, it is not hard to show for moderately large batches $B$, that the number of edges is at least $\frac{\phi}{2}\vol_G(S)$ which suffices for our argument.} Clearly, if half of these edges are still in $(G \setminus B)[V \setminus A]$, then the cut $(S, \bar{S})$ is \emph{not} $\frac{\phi}{2}$-sparse. Otherwise, we have that at least $\frac{\phi}{2} \cdot \vol_G(S) \cdot 8/\phi = 4 \cdot \vol_G(S)$ units of source flow are added to vertices in $S$. By the max-flow min-cut theorem, we further have that since $S$ does not intersect $A$, all of this source flow is routed to sinks by $\ff$. But the total sink capacity of all vertices in $S$ is exactly $\vol_G(S)$ and thus at least $3 \cdot \vol_G(S)$ units of flow are routed out of $S$ by $\ff$. Since no flow is sent into $A$ (by definition of $A$), we thus have that at least $3 \cdot \vol_G(S)$ units of flow are routed through the cut $(S, \bar{S})$ in $(G \setminus B)[V \setminus A]$. And since each edge has capacity at most $8/\phi$, this yields that the number of edges in the cut $(S, \bar{S})$ has to be at least $\frac{3 \cdot \vol_G(S)}{8/\phi} = \Omega(\phi) \cdot \vol_{(G\setminus B)[V \setminus A]}(S)$, as desired.
\end{itemize}

Next, note that the above proof framework extends almost seamlessly to give a low \emph{amortized} recourse algorithm: presented with a sequence of edge deletions $e_1, e_2, \ldots, e_k$ it can maintain $A$ to be monotonically increasing such that after processing the $t$-th deletion, $(G \setminus \{e_1, e_2, \ldots, e_t\})[V \setminus A]$ is a $c \cdot \phi$-expander, for some $c = \Omega(1)$, and $\vol_G(A)$ is at most $O(t/\phi)$. To obtain a proof of this result, it suffices to observe that the flow problem in \cite{saranurak2021expander} evolves such that when an edge is deleted, the maximum amount of flow that was routed via the edge in a previous flow solution is now added as source flow to both endpoints. Together with the fact that edge deletions can only decrease the capacity of cuts, this yields that each source-sink min-cut $A$ is a superset of all previous source-sink min-cuts. This implies the result.

However, in the above framework, the addition of flow caused by a single edge can result in $A$ growing rapidly: a set of large volume that could route barely through a cut might no longer be able to do so after removing just one more edge from the cut causing it to be entirely added to $A$. And while expander pruning has been heavily studied, no previous algorithm has been given that achieves any non-trivial \emph{worst-case} recourse bound.

\paragraph{Flow Certificates for Refined Expander Pruning.} Next, we describe how to use the above algorithm combined with a simple flow certificate to refine previous pruning techniques. To illustrate the power of this new technique, we show that it already yields an \emph{offline} low worst-case recourse algorithm for expander pruning. In this simpler setting, the algorithm can access the full sequence of deletions from the start. Regardless of the simplification, this already constitutes a significant stepping stone toward obtaining a low worst-case recourse algorithm in the dynamic setting. 

Given a $\phi$-expander $G$, we let $e_1, e_2, \ldots, e_k$ be the deletions to $G$ where $e_t$ is the edge deleted from $G$ at time $t$. We denote by $B \subset E$ the set of all these edges, i.e. $B = \{e_1, e_2, \ldots, e_k\}$. Then, our algorithm first instantiates a \cite{saranurak2021expander} flow instance as described in    \Cref{def:overview_flow_instance} for $G$ where set $B$ is deleted. We recall that this flow instance yields a set $A$ of volume $O(|B|/\phi) = O(k/\phi)$ such that $G[V \setminus A] \setminus B$ is a $\frac{\phi}{4}$-expander. 

Given this set $A$, our algorithm initializes a flow $\ff$. This flow $\ff$ is taken to be an arbitrary flow that routes $\deg_G(a)$ units of flow from every vertex $a \in A$ to sinks $b \in V \setminus A$ of capacity $\deg_{G}(b)$ in a network with unit edge capacity $1/\phi$. Such a flow always exists as long as the volume of $A$ is moderately bounded, i.e. $\vol_G(A) \leq \frac{1}{4} |E|$, which is immediate from the definition of expanders and the max-flow min-cut theorem. For the rest of this section, we assume w.l.o.g. that $\ff$ is acyclic, integral, and non-negative in each coordinate. 

We are now ready to describe our offline algorithm: we output the set $\bar{A}$ of pruned vertices and initialize it to be the empty set. Then, at time $t$, when the edge $e_t$ is deleted from $G$, we update the flow certificate $\ff$ as follows: while $\ff(e_t) > 0$, we backtrack the flow $\ff$ from the head of $e_t$ until we find a vertex $s$ that has no incoming flow. Then, we remove one flow unit from this path and update $\bar{A} \gets \bar{A} \cup \{s\}$ if $s \in A$. We repeat this procedure until there is no flow $\ff(e_t)$ anymore and then delete the edge $e$. 

We next show that the above algorithm maintains set $\bar{A}$ with low \emph{worst-case} recourse and such that when pruning it the graph remains expander. 

\begin{claim}
At any time step $t$, the graph $G[V \setminus \bar{A}] \setminus \{e_1, e_2, \ldots, e_t\}$ remains a $\frac{\phi}{12}$-expander and $\bar{A}$ grows by at most $1/\phi$ vertices.
\end{claim}
\begin{proof}
To prove the first part of the claim, consider any cut $S \subseteq V \setminus \bar{A}$ with $\vol_{G}(S) \leq \vol_{G}(V \setminus S)$. We now exploit the following dichotomy: 
\begin{itemize}
    \item \underline{if $\vol_G(S \cap (V \setminus A)) \geq \frac{1}{3} \vol_G(S)$:} then we have that $|E_G(S, V \setminus A)| \geq \frac{\phi}{4} \vol_G(S \cap (V \setminus A)) \geq \frac{\phi}{12} \vol_G(S)$ where the first inequality follows from the fact that $G[V \setminus A] \setminus \{e_1, e_2, \ldots, e_t\} \supseteq G[V \setminus A] \setminus \{e_1, e_2, \ldots, e_k\}$ is a $\frac{\phi}{4}$-expander and a sub-graph of $G$. 
    \item \underline{otherwise (i.e. $\vol_G(S \cap (V \setminus A)) < \frac{1}{3} \vol_G(S)$):} then at most half the $\frac{2}{3} \vol_G(S)$ flow originating from $S \cap A$ can be absorbed by sinks in $S \cap (V \setminus A)$ since they have total capacity $\frac{1}{3} \vol_G(S)$, and therefore at least $\frac{1}{3} \vol_G(S)$ flow gets routed out of $S$ on edges of capacity $1/\phi$. Therefore, $|E_G(S, V \setminus S)| \geq \frac{\phi}{3} \vol_G(S)$. Note that we crucially rely here on the property that we only decrease flow paths that start in vertices with no in-coming flow as it ensures that no vertex is removed (by adding it to $\bar{A}$) that is required by the flow certificate $\ff$.
\end{itemize}

Finally, since the edge capacity of the flow is bounded by $1/\phi$, the set $\bar{A}$ grows by at most $1/\phi$ vertices per edge deletion since we backtrack at most $1/\phi$ flow paths. 
\end{proof}

\paragraph{Expanders are Robust, Even if the Amortized Pruning Suggests Otherwise.} Let us next give a simple online algorithm that gives a natural de-amortization of the algorithm from \cite{saranurak2021expander} for amortized pruning. However, here we make the (strongly) simplifying assumption that the set $A$ as maintained in \cite{saranurak2021expander}, whenever it is augmented by some vertices $\Delta A$ during time $t$, is afterwards not changed by the algorithm during the time steps $t+1, t+2, \ldots, t + \phi \cdot \vol_G(\Delta A)$. That is after every change to $A$, it enters a \emph{resting period}.

To obtain our online algorithm with $O(1/\phi)$ worst-case recourse, we need to leverage a simple but powerful insight: even if the algorithm \cite{saranurak2021expander} adds many vertices to the pruning set $A$ while processing some $t$-th edge deletion $e_t$, just before the update, the graph $G[V \setminus A] \setminus \{e_1, \ldots, e_{t-1}\}$ was still an $\Omega(\phi)$-expander and since expanders are robust, this single edge deletion $e_t$ should only require us to prune out very little.

More concretely, this suggests the following online algorithm: again we output the set $\bar{A}$ of pruned vertices and initialize it to be the empty set. We initialize the algorithm from \cite{saranurak2021expander} to maintain set $A$ with the well-known guarantees. Then, at time $t$, when the edge $e_t$ is deleted from $G$, the algorithm informs the data structure from \cite{saranurak2021expander} which, in turn, updates the set $A$. Let $\Delta A$ denote the vertices added to $A$ during the current time step $t$. 

Since we have that $G[A \setminus \Delta A] \setminus \{e_1, e_2, \ldots, e_{t-1}\}$ is a $c \cdot \phi$-expander for some $c = \Omega(1)$, we know that there is a flow certificate $\ff$ in this graph that routes $\deg_G(a)$ units of flow from every vertex $a \in \Delta A$ to sinks $b \in V \setminus A$ with edge capacities set to $c/\phi$ for some large hidden constant. Once this flow certificate is computed, at time $t$, we update the flow certificate $\ff$ as before: while $\ff(e_t) > 0$, we backtrack the flow $\ff$ from the head of $e_t$ until we find a vertex $s$ that has no incoming flow. Then, we remove one flow unit from this path and update $\bar{A} \gets \bar{A} \cup \{s\}$ if $s \in A$. We repeat this procedure until there is no flow $\ff(e_t)$ anymore and then delete the edge $e$. Further, at any stage $t, t+1, t+2, \ldots, t+\phi \cdot \vol_G(\Delta A)$, we pick for $1/\phi$ rounds, an edge from $E_G(\bar{A}, V \setminus A)$ that still carries flow in $\ff$ (if one exists) and then backtrack one unit of flow on this edge and remove the flow. 

It is not hard to verify that since at each time, we reduce the amount of source of $\ff$ by at least $1/\phi$ (unless there is no source left), and since initially there is only $\vol_G(A)$ units of source flow, we have that at after time $t + \phi \cdot \vol_G(\Delta A)$, the flow $\ff$ is the empty flow. Thus, by our update rule, we have at this point that $\bar{A} = A$. 

Correctness, i.e. that $G[V \setminus \bar{A}] \setminus \{e_1, e_2, \ldots, e_t\}$ remains a $\Omega(\phi)$-expander follows by extending our proof from the last section which uses the flow certificate $\ff$ and the guarantees on $G[V \setminus A] \setminus \{e_1, e_2, \ldots, e_t\}$. 

\paragraph{Relaxing the Resting Property via Batching and Flow Certificate Composition.} While achieving a resting property as assumed above seems unrealistic, it turns out that a straightforward relaxation of this property is already sufficient for our purposes: we request that the amortized algorithm grows the set $A$ such that at any time $t$ it increases by at most $d \cdot 2^i/\phi$ in volume for some $d = \tilde{O}(1)$ and $i$ being the largest integer such that $t$ is divisible by $2^i$. Thus, we allow for a change to $A$ of order $\tilde{O}(2^i / \phi)$ after every $2^i$ time steps. For a current time $t$, let $t' \leq t$ be the largest index such that that $t'$ is divisible by $2^i$, then we define $\Delta A_i$ to be the empty set if $t'$ is also divisible by $2^{i+1}$ and otherwise we define it to consist of the vertices added to $A$ at time step $t'$. Thus, $\Delta A_i$ was added to $A$ before $\Delta A_{i+1}$ for every $i$ and is of size at most $d \cdot 2^i / \phi$.

Next, we give an algorithm that only adds $\tilde{O}(1/\phi)$ vertices to $\bar{A}$ per update to $G$ and ensures that after time $t+2^{i-1}$ where $t$ divisible by $2^i$, we have $\Delta A_i \subseteq \bar{A}$. For our algorithm, we maintain $\lambda = \lceil \log_2(k) \rceil$ flow certificates $\ff_{\lambda}, \ff_{\lambda-1}, \ldots, \ff_0$ where each flow certificate $\ff_i$ is re-computed at any time $t$ divisible by $2^i$ to route from each source $a \in \Delta A_i$ exactly $(\lambda - i+1) \cdot \deg_G(a)$ units of source flow to sinks in $V \setminus A$ where each sink vertex $v$ has sink capacity $\deg_G(v)$ in the current graph $G[V \setminus (\bar{A} \cup \Delta A_{\lambda} \cup \Delta A_{\lambda-1} \cup \ldots \cup \Delta A_{i-1})]$ with sufficiently large edge capacities in $O(1/\phi)$.

A straightforward strategy would now be to delete $1/\phi$ flow paths from each of these flow certificates as proposed in the previous section. However, this process cannot ensure correctness. The problem is that $\ff_{\lambda}$ is no longer a flow certificate when $A$ is changed again because a lower level flow certificate $\ff_i$ might order a vertex $v$ to be added to $\bar{A}$ but through this vertex, we still route flow in the certificate $\ff_{\lambda}$. 

Instead, we maintain the flow certificate $\ff = \ff_{\lambda} + 2 \ff_{\lambda - 1} + 3 \ff_{\lambda - 2} +  \ldots + (\lambda+1) \ff_0$. It is not hard to observe that this flow sends $\deg_G(a)$ units of flow from every vertex $a \in \Delta A_i$ for any $i$. That is because for every such vertex $a$, the certificates $\ff_{\lambda}, \ff_{\lambda - 1}, \ldots, \ff_{i+1}$, each route at most $\deg_G(a)$ units of flow into $a \in \Delta A_i$ as it appears as a sink in these flow problems, but since $\Delta A_i$ is in $A$ at any stage that $\ff_i, \ff_{i-1}, \ldots, \ff_0$ are computed and thus it does not appear as a sink in these problems. Thus, the amount of flow sent to $a$ by all certificates is at most $(\lambda - i) \deg_G(a)$ units. But since $\ff_i$ sends away $(\lambda - i + 1)\deg_G(a)$ units of flow from $a$, we have that the net flow leaving $a$ is at least $\deg_G(a)$ in $\ff$. Further, each edge carries at most $c' = O(\lambda/\phi)$ flow in $\ff$. For simplicity, we assume for the rest of the section that $\ff$ is acyclic (while this is not true w.l.o.g., a similar property that suffices for our algorithm does hold).

We update the flow certificate $\ff$ similar to how we updated the certificate in the previous section. However, slightly more care is required and we end up with the following update rule: at each time $t$, we first update $\ff$ by removing all flow from the edge $e_t$ that is deleted at the current time from $G$. Again, we do so by backtracking the flow carefully to a vertex that has no in-flow. And thereafter, for $4d c' \cdot \lambda$ rounds (where $d$ bounds the recourse on $A$ and $c'$ is the maximum amount of flow on any edge in $\ff$), we find the smallest index $i$, such that $\Delta A_i \not\subseteq \bar{A}$, then find an edge $e \in E_G(\Delta A_i, V \setminus \bar{A})$ that still carries flow and backtrack the flow on this edge (note that all flow in $\ff$ leaves $\Delta A_i$ via these edges because of the max-flow min-cut theorem). Vertices $a \in A$ with less than $\deg_G(a)$ out-flow in $\ff$ are then added to $\bar{A}$ in the final step of the algorithm to process update $e_t$.

To establish correctness, observe that we reduce the amount of flow $\ff$ that leaves $\Delta A_i$ in each such round by at least $1$ unit, and the total amount of flow on $\Delta A_i$ is at most $c' \cdot \vol_G(A_i) \leq c' \cdot d\cdot 2^i$. It is thus not hard to show that for every index $i$, the flow $\ff$ does not have $\Delta A_i$ in its support after $2^{i-1}$ time steps
since the computation of $\ff_i$ have passed and therefore re-setting $\Delta A_i$ is safe. 

We can thus conclude that $\ff$ is indeed a flow certificate together with the graph $G[V \setminus \bar{A}] \setminus \{e_1, \ldots, e_t\}$. Since the maximum amount of flow on each edge in $\ff$ is larger by a factor $\lambda+1$ compared to the flow in the last section and each sink vertex might receive up to $\lambda+1$ times its degree units of flow, we can now only argue that it proves that $G[V \setminus \bar{A}]$ is a $\Omega(\phi/\log^2 m)$-expander since $\lambda = O(\log m)$. The recourse can be bounded by $\tilde{O}(1/\phi)$. See \Cref{fig:cut_detection} for an illustration of nesting flow certifcates and the detection of sparse cuts using \cite{saranurak2021expander}-flow instances.

\begin{figure}[h]
    \centering
    \includegraphics[width = 10cm]{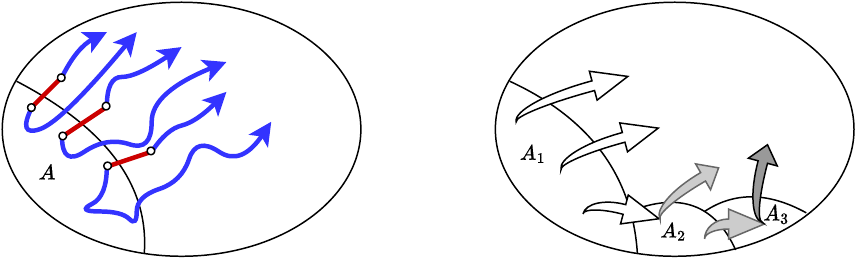}
    \caption{On the left: After deleting the red edges and adding $8/\phi$ source to their endpoints, it is still possible to route the flow to sinks $v \in V$ of capacity $\deg_G(v)$. If this is the case the graph is guaranteed to be a $\phi/10$-expander. Otherwise, a sparse cut $A$ is detected by the algorithm and thereafter the remaining graph $G[V \setminus A]$ is guaranteed to be a $\phi/10$-expander. On the right: Vertex sets $A_1, A_2$ and $A_3$ are yet to be pruned from the graph. Each of these maintains a flow certificate, such that adding them up ensures that the remaining graph is a $\tilde{\Omega}(\phi)$-expander.}
    \label{fig:cut_detection}
\end{figure}

\paragraph{An Amortized Online Algorithm with Resting Periods via Increasing Sink Capacities.} Finally, we still need an amortized algorithm that satisfies our relaxed resting period property: we request that the amortized algorithm grows the set $A$ such that at any time $t$ it increases by at most $\tilde{O}(2^i/\phi)$ in volume for $i$ being the largest integer such that $t$ is divisible by $2^i$. 

A natural approach is to use batching techniques: at any time step $t$ we compute $i$, the largest integer such that $t$ is divisible by $2^i$, and then compute a maximum flow $\ff_i$ for the \cite{saranurak2021expander}-flow instance as described in \Cref{def:overview_flow_instance} for the remaining graph $G_i \defeq G[V \setminus (\bar{A} \cup \Delta A_{\lambda} \cup \Delta A_{\lambda-1} \cup \ldots \cup \Delta A_{i+1})]$ and deletion batch $B_i$ containing all the edge deletions that occurred since $\Delta A_{i+1}$ was last updated, i.e. the last $2^i$ edge deletions. As shown directly after the statement of \Cref{def:overview_flow_instance}, such a flow instance yields a set $\Delta A_i$ that has volume at most $\tilde{O}(2^{i}/\phi)$, as well as the guaranteeing that $G_{i + 1}$ is a $\frac{\phi}{4}$-expander \emph{assuming} that $G_i$ is a $\phi$ expander. Furthermore, the set $\Delta A_i$ naturally only gets updated periodically every $2^i$ deletions. 

Unfortunately, the constant loss in expansion per level is not acceptable. Given that $G_\lambda$ has expansion $\phi$, the graph $G_i$ is only guaranteed to be a $\frac{\phi}{4^{\lambda -i}}$-expander, and thus the graph $G_0$ might only have expansion $\frac{\phi}{4^\lambda}$. In particular, for our previous choice $\lambda = O(\log m)$, this only yields the trivial bound $\frac{\phi}{\poly(m)}$ on the expansion of $G_0$. A possible remedy is to drastically decrease the number of levels $\lambda$ to $O(\sqrt{\log m})$. Even though this allows non-trivially bounding the expansion of $G_0$, it introduces an undesirable sub-polynomial loss of $2^{O(\sqrt{\log m})}$ in expansion and growth of the pruned set. 

To recover the desired $\Omega(\phi)$ bound using $\lambda = O(\log m)$ levels, we re-visit the idea of maintaining a single \cite{saranurak2021expander}-flow instance as described in \Cref{def:overview_flow_instance}. When describing the amortized recourse pruning algorithm, we argued that attempting to route a single extra particle of flow on the residual graph might saturate a large cut in the maximum flow instance, and therefore dis-proportionally blow up the size of the set $\Delta A_i$. We next describe a simple technique that remedies this issue. We point out that here, we merely extend a technique which previously appeared in \cite{chen2025parallel}, where it was exploited to parallelize \cite{saranurak2021expander}. 

Recall that given a set of deletions $B \subseteq E$ to a graph $G = (V, E)$ that was initially a $\phi$-expander, a \cite{saranurak2021expander}-flow instance places $\frac{8}{\phi}$ source capacity at vertex $v$ per edge $(u,v)$ in $B$ with endpoint $v$. It sets the sink capacity of every vertex $v \in V$ to $\deg_G(v)$, as well as setting all the edge capacities to $8/\phi$. Then we let $\ff$ denote a maximum flow for this instance, and we again let $(A, V \setminus A)$ be the most balanced cut for which every edge is saturated with flow from $A$ to $V \setminus A$. Then $\vol_G(A) \leq 16 |B|/\phi$ because all the sinks in $A$ are saturated, and $G[V \setminus A]$ remains a $\frac{\phi}{4}$-expander. Now, assume that we are given an update $\Delta B \subseteq E$ consisting of additional edge deletions and want to compute an updated flow $\ff' = \ff + \Delta \ff$, and set $A'$ for the deletion set $B' = B \cup \Delta B$ such that the volume of $\Delta A$ is bounded by $O(|\Delta B|/\phi)$ where $\Delta A = A' \setminus A$.

A simple way to update the flow $\ff$ is to add $\frac{8}{\phi}$ units of source at vertex $v$ for every edge $(u, v)$ in $\Delta B$ with endpoint $v$ and computing a maximum flow $\Delta \ff$ on the residual graph $\dir{G}_{\ff}[V \setminus A]$ after removing the edges in $\Delta B$. Adding $\frac{8}{\phi}$ source to every endpoint ensures that at most $\frac{16}{\phi}$ additional source is routed away per edge adjacent to any $v$ in $\Delta B$, and that the flow $\ff$ is repaired since very deleted edge carried at most $\frac{8}{\phi}$ flow. Therefore, $\ff' = \ff + \Delta \ff$ is a suitable flow. 

To control the size of $\Delta A$, we make the following crucial observation: Adding only $\frac{1}{2} \deg_G(v)$ sink capacity to every vertex initially only worsens the volume bound on $A$ by a factor $2$. Then, we can afford to add an \emph{additional} $\frac{1}{2} \deg_G(v)$ sink capacity to the residual problem after receiving the update $\Delta B$. But this additional sink capacity has to be saturated by every vertex in $\Delta A$, and we therefore obtain that the volume of $\Delta A$ is bounded by $O(|\Delta B|/\phi)$. 

Given this insight, it is easy to extend the scheme to batches $B_\lambda, \ldots, B_0$ arriving one-by-one: We initially set the sink capacity to $\frac{1}{\lambda + 1} \deg_G(v)$, and then increase it by $\frac{1}{\lambda + 1} \deg_G(v)$ when processing each individual batch. This only worsens our volume bound by a $\lambda + 1 = O(\log m)$ factor, while ensuring that the expansion is still bounded by $\Omega(\phi)$. 

\paragraph{Fast Flow Backtracking.} Finally, to turn the ideas presented so far into a fast algorithm, we need to compute and backtrack the flows very efficiently. Setting aside the issue of computing flow certificates, let us first discuss how to backtrack flows efficiently.

To this end, we propose the following simple algorithm: initialize a rooted, directed forest $T$ to be the empty forest. Whenever we remove flow from an edge $(u, v)$, we search for the first edge on the path from $u$ to its root in $T$ that carries $0$ flow if it exists. We remove this edge from $T$. Then, we search for the new root $r$ of $u$ in $T$, and remove $1$ flow from the corresponding path. Finally, we check if there is an edge $e$ adjacent to $r$ carrying non-zero flow into $r$. If such an edge $e = (w, r)$ exists, we add it to $T$. 

Although the above algorithm sometimes backtracks to a vertex that is not a true source, such a vertex $a$ can be erroneously backtracked to at most $\deg_G(a)$ times. If we initially send out an extra $\deg_G(a)$ flow for every vertex $a \in \Delta A_i$, and then only add $a$ to $\bar{A}$ once at least $\deg_G(a)$ units of source were removed from $a$, we maintain that the remaining flow is still an adequate flow certificate. Furthermore, backtracking on such a tree is easy to implement using link-cut trees \cite{ST83}, and this backtracking procedure nicely composes across levels. 

\paragraph{Fast Flow Computations.} Let us next discuss how to make the overall algorithm efficient. We start by discussing how to implement the flow computations efficiently. For now, we focus on obtaining fast amortized update times.  

Since running max-flow algorithms requires us to read in the entire graph $G$, and we often need to compute flow much faster, we turn to the idea already present already in  \cite{orecchia2014flow} to obtain fast \emph{local} flow algorithms. In particular, all the amortized pruning and flow certificates are computed on (almost-)expander graphs and thus can be computed using only very few iterations of Dinitz blocking flow algorithm. Since Dinitz algorithm can be implemented \emph{locally}, we can then argue that the runtime of each flow computation only scales in the support size of the flow.

However, we point out that this leads to various complications in the arguments. The above flow algorithms are only approximate and this requires great care when composing multiple levels in our schemes. However, we overcome these technical difficulties.

Finally, we also need to de-amortize flow computations. Here, we turn to classic de-amortization techniques: we compute the flows for pruning and flow certificates in the background and then only use them once the computation has finished. But at this point $G$ might have evolved and some of the edges that the flow uses no longer be present in $G$. While this is ok for the pruning algorithm, it is not for our flow certificates. But to remedy that flow certificates are no longer up-to-date, we can again use backtracking to repair certificates by removing flow that is on edges no longer in $G$. Since this again might take some time to finish, we might then have to iterate the scheme up to $O(\log m)$ times before completing. As the reader might already suspect, here we sweep many technical details, that require careful implementation, under the rug. 

\paragraph{Expander Pruning with Worst-Case Update Time and Recourse.} The above paragraph completes our description of the algorithm. We summarize briefly: over algorithm internally runs an amortized algorithm that prunes with a new resting property that allows us to only think about $O(\log m)$ batches of pruned sets at any time. We use flow certificates to de-amortize this algorithm as they allow us to explicitly check which vertices can still be pruned later and which require urgent attention. Finally, to de-amortize the runtime of both of these algorithms, we develop new techniques to backtrack flow on flow certificates and design a careful hierarchy to compose flow certificates and to de-amortize flow computations via background computations.

Together, this new framework yields \Cref{thm:double_low_recourse}, our main result.

\section{Preliminaries}

\paragraph{Vectors and matrices. } We denote vectors as bold lowercase letters $\xx$. For $\xx \in \R^n$ and $S \subseteq [n]$ we let $\xx(S)$ denote the sum of the entries of $\xx$ in the coordinates of $S$, i.e. $\xx(S) = \sum_{i \in S} \xx(i)$, and we let $\xx[S] \in \R^n$ denote $\xx$ restricted to $S$, i.e. $\xx[S](i) = \xx(i)$ if $i \in S$ and $\xx[S](i) = 0$ otherwise. We use $\mathbf{0}$ and $\mathbf{1}$ to denote the all-$0$ and the all-$1$ vectors respectively. We use $\mathbf{1}_u$ to denote the indicator vector which is zero everywhere but for coordinate $u$ where it is equal $1$. We use $\supp(\xx)$ to denote the set of coordinates of $\xx$ whose value is non-zero.

\paragraph{Graphs.} In this article, we work with both undirected and directed graphs, uncapacitated and capacitated. An undirected capacitated graph $G=(V,E,\cc)$ consists of a vertex set $V$, an edge set $E$ where each edge is represented as a two-element set, i.e. if there is an undirected edge between $u$ and $v$ in $G$, then $\{u,v\} \in E$, and finally we have capacities $\cc \in \R^{|E|}_{\geq 0}$. We let $\mathbf{deg}_G$ denote the degree vector of graph $G$ where for each vertex $v \in V$, we have $\mathbf{deg}_G(v)$ equal to the number of edges incident to $v$. We denote by $\vol_G(S)$ for any $S \subseteq V$, the sum of degrees of vertices in $S$. We denote by $E(A, B)$ for any $A, B \subseteq V$ the set of edges in $E$ with one endpoint in $A$ and another in $B$. Further, we let $\partial_G(S)$ be the set of edges between $S$ and $V \setminus S$ in $G$. 

For a directed graph $\dir{G} = (\dir{V}, \dir{E}, \dir{\cc})$, we have $\dir{V}$ again being the vertex set, and $\dir{E}$ be a set of two-tuples, where $(u,v) \in \dir{E}$ if and only if there is an edge directed from $u$ to $v$. We write $\dir{G} \gets \textsc{Directify}(G)$ to denote a transformation of an undirected graph $G$ into a directed (eulerian) graph where $\dir{V} = V$, $\dir{E} = \{ (u,v) | \{u,v\} \in E\}$ and each such tuple $(u,v)$ receives capacity $ \dir{\cc}(u,v) = \cc(\{u,v\})$. We let $\deg^{\text{out}}_{\dir{G}}$ be the out-degree vector of $\dir{G}$, i.e. the for vertex $v$, $\mathbf{deg^{out}}_{\dir{G}}(v)$ is the number of edges in $\dir{E}$ with tail $v$. We denote by $\dir{E}(A, B)$ for any $A, B \subseteq V$ the set of edges in $E$ with tail in $A$ and head in $B$. Given any set $S \subseteq V$, we let $\dir{G}[S]$ denote the graph induced by the vertex set $S$.

For convenience, we extend the capacity functions $\cc$/ $\dir{\cc}$ to all two-element sets $\{u,v\} \subseteq V$/ tuples $(u,v) \in V^2$ to be zero everywhere where it is not defined.

\paragraph{Expanders. } We call a graph a $\phi$-expander (with respect to conductance) if, for every non-empty set $S \subset V$ where $\vol_G(S) \leq \vol_G(V \setminus S)$, we have $E(S, V \setminus S) \geq \phi \vol_G(S)$. We define expansion with respect to the starting volume for decremental graphs, which will be the suitable measure throughout this article.

\begin{definition}[Decremental Expander Graph]
    \label{def:pruned_exp}
    We call a graph $G = (V, E)$ possibly undergoing a sequence of edge and vertex deletions, a $\phi$-expander, if for all $S \subset V$ so that $\vol_{G^{(0)}}(S) \leq \vol_{G^{(0)}}(S \setminus V)$ we have $|E(S, V \setminus S)| \geq \phi \vol_{G^{(0)}}(S)$ where $G^{(0)}$ denotes the initial graph $G$ before any deletions happened.
\end{definition}

Notice that measuring with respect to the starting volume implies that a decremental expander graph $G$ is also an expander with respect to conductance because the starting volume of a set is an over-estimate of the current volume of said set. 

\paragraph{Flows.} For flows, we only consider directed, capacitated graphs $\dir{G} = (\dir{V}, \dir{E}, \dir{\cc})$. We call a (pre-)flow $\ff \in \R^{\dir{E}}$ feasible if $\mathbf{0} \leq \ff \leq \dir{\cc}$. We let $\dd \in \R^{V}$ denote a demand vector, where positive entries correspond to sources and negative entries to sinks. We consider demands with non-positive sum, i.e. where we have at least as much sink capacity as source capacity. 

We let $\BB_{\dir{G}} \in \R^{V \times \dir{E}}$ be the edge-vertex incidence matrix 
\begin{align*}
    \BB_{\dir{G}}(v, e) = \begin{cases} 
        -1 & \text{if } e = (u,v) \\
        1 & \text{if } e = (v,u) \\
        0 & \text{otherwise}
    \end{cases}
\end{align*}
Observe that $(\BB_{G} \ff)(v)$ is precisely the net flow going out of $v$.
Let $\ss \in \R^V$ be a \emph{source function} which indicates, for each vertex $v$, the amount of flow initialized at $v$. 
Let $\tt \in \R^V_{\ge 0}$ be a \emph{sink function} which indicates, for each vertex $v$, the sink capacity that $v$ can absorb. 
Given a flow $\ff$ in $G$, we define the \emph{excess} $\xx_{\ss, \tt, G, \ff} = \max(\ss - \tt - \BB_{G} \ff, \veczero) \in \R^V$. 
Intuitively, $\xx_{\ss, \tt, G, \ff}(v)$ measures the total amount of flow at $v$ that is not absorbed. 
We omit some parameters in the subscript whenever it does not cause confusion.
If $\xx_{\ss, \tt, G, \ff} = \veczero$, then we say that \emph{$\ff$ routes source $\ss$ to sink $\tt$}.

For a subset $S$ of the vertices, we let the amount of source in $S$ refer to $\ss(S)$, i.e. the sum of the source values on vertices in $S$. 

We define the residual graph to be $\dir{G}_{\ff} = (\dir{V}, \dir{E}_{\ff}, \dir{\cc}_{\ff})$ where we have $\cc_{\ff}(u,v) = \max\{0, \cc(u,v) - \ff(u,v) + \ff(v,u)\}$ for all $(u,v) \in V^2$, and $(u,v)$ in $\dir{E}_{\ff}$ if and only if $\cc_{\ff}(u,v) > 0$.

\paragraph{Dinitz's Algorithm.} The blocking flow algorithm by Dinitz yields the following guarantee for our specific setting, which was first observed in \cite{orecchia2014flow}.

\begin{fact}[Local Blocking Flow, See \cite{orecchia2014flow}]
    \label{fct:dinitz}
    Given a graph $\dir{G} = (\dir{V}, \dir{E}, \dir{\cc})$, demands $\ss - \tt = \dd \in \R^{V}$ with $\sum_i \dd(i) \leq 0$ for $\ss \in \R^V, \tt \in \R^V_{\geq 0}$, and a parameter $h \in \mathbb{N}_{> 0}$. Then, running Dinitz Algorithm for $h$ rounds yields a feasible pre-flow $\ff$ such that no path consisting of less than $h$ edges exists from any vertex $u$ with $\xx_{\ff}(u) > 0$ to any vertex $v$ with $\dd(v) + [\BB \ff](v) < 0$ in the residual graph $\dir{G}_{\ff}$.  
    If $\tt \geq \deg_G/\lambda$ for some $\lambda \geq 1$, then Dinitz algorithm can be implemented in time $O(\lambda \cdot h \cdot \norm{\ss}_1)$.
\end{fact}
\begin{proof}
    We say a vertex is explored if the blocking flow sends flow to at least one of its out edges. Whenever a vertex $v$ is explored (which can happen up to $h$ times) it absorbs at least $\deg_G(v)/\lambda$ flow, which can pay for these explorations. 
\end{proof}

\section{Expander Pruning via Batching}
\label{sec:amortized}

We first present the foundation of our algorithm based on the batch pruning algorithm of \cite{saranurak2021expander}. It's runtime guarantees are also amortized, but unlike \cite{saranurak2021expander} it processes batches $B_1, \ldots, B_{k}$ for $k = \ceil{\log_2 m}$ in time \emph{proportional} to $|B_i|$ up to poly-logarithmic factors, where we choose batches that partition the set of edges deleted from $G$ up to the current time such that batch $B_i$ consists of $O(m/2^i)$ such deleted edges (later, we will additionally ensure that the pruned set only grows but it is useful to first ignore this technicality). This is achieved by extending a technique that was previously used in \cite{chen2025parallel} to parallelize \cite{saranurak2021expander}. 

Because the runtime per batch is roughly proportional to its size, standard de-amortization techniques could directly yield an algorithm with worst-case runtime guarantees. Unfortunately, algorithms derived in this manner rely on pre-computing information in the background and occasionally switching pointers to the current answer. This method creates large recourse. In particular, another algorithm using the data structure can only read all the changes in fast amortized time when they are published. 

We address this caveat in the following two sections. In \Cref{sec:amortized_low_rec} we show that low-recourse can be achieved in our framework, albeit still with amortized run-time. To the best of our knowledge, this is the first proof that low recourse expander pruning is structurally possible. Then, in \Cref{sec:worst_case_low_rec} we use relatively involved de-amortization techniques to improve this runtime to be poly-logarithmically bounded in the worst case. 

\begin{lemma}[Expansion Certificate, See Proposition A.3 in \cite{saranurak2021expander}] \label{lem:exp_cert}
    Let $G = (V, E)$ be a $\phi$-expander, $A \subset V$ and $B \subset E$. Let $G' = G[V \setminus A] \setminus B$.
    If there exists a flow $\ff$ routing $(\frac{8}{\phi}(\deg_G - \deg_{G'}))[V \setminus A]$ source to sinks $2\deg_G[V \setminus A]$ in $G'$ with edge capacities $\frac{32}{\phi}$, then $G'$ is a $\phi/10$ expander. 
\end{lemma} 
\begin{proof}
    To arrive at a contradiction, consider a cut $S$ that is $\phi/10$ sparse in $G'$, i.e., $|E_{G'}(S, V(G') \setminus S)| < \phi/10 \cdot \vol_G(S)$
    We know that $|E_G(S, V \setminus S)| \geq \phi \vol_G(S)$ as $G$ is a $\phi$-expander. Therefore, the amount of source in $G'[S]$ is at least $\frac{8}{\phi}|E_G(S, V \setminus S)| - |E_{G'}(S, V(G') \setminus S)| \geq \frac{9\cdot 8}{10} \vol_G(S) = \frac{72}{10} \vol_G(S)$. But there are less than $\frac{\phi}{10}\vol_G(S)$ edges leaving the cut $S$ in $G'$, and therefore the total capacity of these edges is at most $\frac{32}{10}\vol_G(S)$ which yields a contradiction since at least $\frac{72}{10}\vol_G(S) - 2\vol_G(S) > \frac{32}{10}\vol_G(S)$ flow has to leave the set $S$ as only $\vol_G(S)$ of the flow can be absorbed by sinks in $S$.
\end{proof}

\subsection{The Meta-Algorithm: Batch Pruning}

For a given number of levels $k = \ceil{\log_2(m)}$ and $\lambda \defeq k + \ceil{\log_2(\phi^{-1})} + 1 \leq 4 \log_2 m$, the algorithm \textsc{BatchPruning}($G = (V, E), \{B_i\}_{i = 1}^\lambda$, $\{A_i\}_{i = 1}^\lambda$, $\phi$) (\Cref{alg:batch_pruning}) processes a sequence of batches $\{B_i\}_{i = 1}^{\lambda}$, where each batch $B_i$ consists of either $0$, $2^{k - i - 1}$ or $2^{k - i}$ edges to be deleted from a $\phi$-expander $G$ (the auxiliary batches $B_{k + 1}, \ldots B_{\lambda}$ are always empty).

It is additionally given a sequence $\{A_i\}_{i = 1}^\lambda$, where each $A_i \subseteq V$ consists of a set of vertices. These sets correspond to previously pruned vertices and are used to ensure the monotonicity of the pruned set. In the end, the algorithm returns a set $\widehat{V}$ so that $G' = G[\widehat{V}] \setminus \bigcup_{i = 1}^k B_i$ is still a $\Omega(\phi)$-expander and $\widehat{V} \cap A_i = \emptyset$ for all $i \in [k]$, i.e. all the vertices in the sets $\{A_i\}_{i = 1}^\lambda$ are pruned. The reader might want to assume they are all empty for an initial read. 

Throughout, the algorithm maintains sources $\ss$, sinks $\tt$ as well as a flow $\ff$ with edge capacity $8/\phi$ that routes part of the sources to some sinks in a subgraph $\widehat{G}$ of $G$. Our flow problems always attempt to route source $\ss$ to sinks $\tt$. However, for algorithmic efficiency, the flow problems are phrased in terms of a residual graph $\widehat{G}_{\ff}$ over the course of the algorithm. To ensure correctness of our algorithm, whenever an edge $e$ is deleted from $G$, we add $\BB \ff[e]$ to some additional demand vector $\ddelta$ after every edge deletion. 

We now describe the algorithm in more detail (see \Cref{alg:batch_pruning} for pseudocode). Initially $\widehat{G} \gets G$. Since the flow $\ff$ does not route all the demands, there may be some excess source (and sink) capacity left at some vertices over the course of the algorithm. The main loop at \Cref{alg:batch_pruning:main_loop} of \Cref{alg:batch_pruning} removes batch $B_i$ from the graph $\widehat{G}$, and then adds $8/\phi$ units of source to both endpoints of each removed edge in $\ss$. Then, it also removes the set $A_i$ from $\widehat{G}$ and again adds $8/\phi$ source to endpoints of removed edges. Finally, it adjusts the extra demand vector $\ddelta$ values of the endpoints with regards to the flow on the removed edge by adding the net demand of the removed edges to $\ddelta$. It finally removes the set of vertices in $S_i$ from the graph, and adds $8/\phi$ source to every endpoint of an edge in $\widehat{E}(S_i, \widehat{V} \setminus S_i)$ and again adjusts the source for flow going over that edge as before by reducing the source by the amount of flow carried from $S_i$ to $\widehat{V} \setminus S_i$, or increasing it if the flow went in the other direction. 

Furthermore, in every iteration, it adds $\frac{1}{\lambda} \cdot \deg(v)$ sink capacity to every vertex $v$, where we recall that $\lambda \defeq k + \log_2(1/\phi) + 1$ which yields a very structured flow problem.\footnote{In the following, we assume that all degrees are multiples of $\lambda$. This can be achieved by replacing every edge with $\lambda$ multi-edges, and causes a single $O(\log n)$ factor overhead.} Then, it reduces the amount of excess by first running a few iterations of Dinitz's blocking flow algorithm, which can be implemented very quickly given the structure of the problem. If the excess source is not reduced enough, we show that a cut whose removal reduces the source exists and is found. Once the excess is reduced enough, the next batch is processed. See the pseudocode presented in \Cref{alg:batch_pruning} for a detailed description of our algorithm. 

\begin{algorithm}[ht]
\caption{\textsc{BatchPruning}($G = (V, E), \{B_i\}_{i = 1}^{\lambda}, \{A_i\}_{i = 1}^{\lambda} , \phi$)}
\label{alg:batch_pruning}
\DontPrintSemicolon
$\widehat{G} = (\widehat{V}, \widehat{E}, {\cc}) \gets (V, E, \cc = \frac{8}{\phi} \cdot \vecone$); $\ss_0 \gets \veczero$; $\tt_0 \gets \veczero$; $\ff_0 \leftarrow \veczero$; $\ddelta_0 \gets \veczero$ \\
\For{$i = 1, \ldots , \lambda \defeq k + \log_2(\phi^{-1}) + 1$}{ \label{alg:batch_pruning:main_loop}
    $\widehat{G}_{\text{new}} = (\widehat{V}_{\text{new}}, \widehat{E}_{\text{new}}) \gets \widehat{G}[\widehat{V} \setminus A_i] \setminus B_i$ \tcp*{Compute next graph}
    $\ss_{i} \gets (\ss_{i - 1} + \frac{8}{\phi}(\deg_{\widehat{G}} - \deg_{\widehat{G}_{\text{new}}}))[\widehat{V}_{\text{new}}]$ \label{alg:batch_pruning:source_start} \tcp*{Add source values for $\widehat{G}_{\text{new}}$}
    $\ddelta_{i} \gets \ddelta_{i - 1}[\widehat{V}_{\text{new}}] + (\BB \ff_{i - 1})[\widehat{V}_{\text{new}}] - \BB_{\widehat{G}} \ff_{i - 1}[\widehat{E}_{\text{new}}]$ \\
    $\tt_i \gets \tt_{i - 1}[\widehat{V}_{\text{new}}]$ \\
    $\widehat{G} \gets \widehat{G}_{\text{new}}$ \label{alg:batch_pruning:graph} \\ 
    $\tt_i \gets \tt_i + \frac{1}{\lambda} \cdot \deg_{G}[\widehat{V}]$ \tcp*{Add extra sink for this round}
    $\ff_i \gets \ff_{i - 1}[\widehat{E}] + \textsc{Dinitz}({\widehat{G}}_{\ff_{i - 1}[\hat{E}]}, \ss_i - \tt_i - \BB \ff_{i - 1}, h = \frac{C \cdot \lambda \cdot \log_2 n \log_2 \log_2 n}{\phi})$ \label{alg:batch_pruning:f} \tcp*{$C \defeq 10^8$}
    \uIf(\tcp*[f]{Recall $\xx_{\ss_i, \tt_i, \ff_i} \defeq \max(\ss_i - \tt_i - \BB \ff_i, \veczero)$}){$\norm{\xx_{\ss_i, \tt_i, \ff_i}}_1 > 2^{k-i}/\phi $}{ \label{alg:batch_pruning:if} 
        $j \gets 0$; $S_{\leq 0} \gets \text{supp}(\xx_{\ss_i, \tt_i, \ff_i})$ \\
        \While{$|\widehat{E}_{\ff_i}(S_{\leq j}, \widehat{V} \setminus S_{\leq j})| \geq \frac{\phi}{10^6 \cdot \lambda \cdot \log_ n} (\vol_{G}(S_{\leq j}) + \norm{\xx_{\ss_i, \tt_i, \ff_i}}_1)$}{ \label{alg:batch_pruning:while}
            $j \gets j + 1$; $S_{\leq j} \gets \{v \in \widehat{V}: \dist_{\widehat{G}_{\ff_i}} (S_{\leq 0}, v) \leq j \}$ \\
        }
        $S_i \gets S_{\leq j}$ \label{alg:batch_pruning:proposal} \\
        $\ss_i \gets \ss_i[\widehat{V} \setminus S_i] + \frac{8}{\phi} \cdot (\mathbf{deg}_{\widehat{G}} - \mathbf{deg}_{\widehat{G}[\widehat{V} \setminus S_i]})$ \label{alg:batch_pruning:src_update}\\
        $\ddelta_i \gets \ddelta_i + (\BB \ff_i)[\widehat{V} \setminus S_i] - \BB \ff_i[E(\widehat{G}[\widehat{V}])]$ \\
        $\tt_i \gets \tt_i[\widehat{V} \setminus S_i]$ \label{alg:batch_pruning:st_end} \\
        $\widehat{G} \gets \widehat{G}[\widehat{V} \setminus S_i]$; \label{alg:batch_pruning:G_end} \\
        $\ff_i \gets \ff_i[\widehat{E}]$ \label{alg:batch_pruning:f_end} 
    }
    }
    \Return{$\widehat{V}$}
\label{alg:batchPruningAlg}
\end{algorithm}

\subsection{Analysis of Batch Pruning} 

In this section, we prove various properties of the batch pruning algorithm in a series of claims. These will set us up for the proving the main result of this chapter: A dynamic pruning algorithm that is directly suitable for de-amortization. We first show that the while loop at \Cref{alg:batch_pruning:while} terminates in less than $h$ iterations whenever it is called and the excess is bounded. 

\begin{claim}
    \label{clm:while_terminates}
    Assume that the total excess after \Cref{alg:batch_pruning:graph} of \Cref{alg:batch_pruning} is at most $\norm{\xx_{\ss_i, \tt_i, \ff_{i - 1}}}_1 \leq (64 \log_2 n) \cdot 100 \cdot 2^{k - i} / \phi$, and that $\norm{\ss_i}_1 \geq \norm{\tt_i}_1$. Then, the while loop at \Cref{alg:batch_pruning:while} terminates after $j < h$ iterations if it is entered. 
\end{claim}
\begin{proof}
    Using the if condition ($\norm{\xx_{\ss_i, \tt_i, \ff_i}}_1 > 2^{k-i}/\phi $), we have that every iteration of the while loop adds at least $2^{k - i}/(10^6  \cdot \lambda \cdot \log_2 n)$ to $\vol_G(S_{\leq j})$. Therefore, after $r_1 = \frac{5 \cdot 10^7 \cdot \lambda \log_2 n}{\phi}$ iterations, we have $\vol(S_{\leq r_1}) \geq 2^{k - i}/\phi$. 
    
    Since the volume $\vol(S_{\leq j})$ also grows multiplicatively by $\left(1+\frac{\phi}{10^6 \cdot \lambda \cdot \log_2 n}\right)$, we have after $r_1 + r_2$ iterations for any $r_2$ that $\vol(S_{\leq r_1 + r_2}) \geq \left(1 + \frac{\phi}{10^6 \cdot \lambda \cdot \log_2 n}\right)^{r_2} 2^{k - i}/\phi$. Fixing $r_2 = \frac{5 \cdot 10^7 \cdot \lambda \cdot \log_2 n \log_2 \log_2 n}{\phi}$ and thus $r_1 + r_2 \leq h$ this yields $\vol_G(S_{\leq r_1 + r_2}) > 10^5 \cdot \log_2^2(n) \cdot 2^{k - i}/\phi $ since $(1 + \epsilon)^{1/\epsilon} \geq 2$ for $\epsilon \leq 1$ and therefore $\vol_G(S_{\leq r_1 + r_2}) \geq 2^{50 \log_2 \log_2 n} \cdot 2^{k - i}/\phi \geq 10^5 \log_2^2 n 2^{k - i}/\phi$. But since the shortest path from a sources to sinks in the residual graph is at least $h$ by \Cref{fct:dinitz}, all the sinks inside $\vol(S_{\leq h})$ are saturated. 
    Thus, $\vol(S_{\leq h}) \leq (64 \log_2 n \lambda) \cdot 100 \cdot 2^{k - i}/\phi \leq 512 \log_2^2n  \cdot 100 \cdot 2^{k - i}/\phi$ by the assumption on the initial excess, and therefore we obtain a contradiction. This proves that the while loop terminates for some $j < h$. 
\end{proof}

\begin{claim}\label{clm:source_reduction}
Assume that the total excess after \Cref{alg:batch_pruning:graph} of \Cref{alg:batch_pruning} is at most $\norm{\xx_{\ss_i, \tt_i, \ff_{i - 1}}}_1 \leq (64 \log_2 n) \cdot 100 \cdot 2^{k - i} / \phi$, $A_i = \emptyset$, and that $\norm{\ss_i}_1 \geq \norm{\tt_i}_1$. Then, at the end of the $i$-th iteration of the for loop in \Cref{alg:batch_pruning:main_loop} of \Cref{alg:batch_pruning}, we have $\norm{\xx_{\ss_i, \tt_i, \ff_{i} }}_1 \leq 2^{k - i}/\phi$. Furthermore, $\vol_G(S_i) \leq (64 \log_2 n) 100 \cdot \lambda \cdot 2^{k - i}/\phi.$
\end{claim}
\begin{proof}
    If the condition at \Cref{alg:batch_pruning:if} evaluates to false, then the claim immediately follows. Otherwise, the while-loop at \Cref{alg:batch_pruning:while} terminates after less than $h$ iterations by \Cref{clm:while_terminates}. 

    Then, we notice that every edge $e = (u,v)$ in $E_G(S_i, \widehat{V} \setminus S_i)$ that is not leaving $S_i$ in the residual graph is saturated with flow. Therefore, the adjustment in $\ddelta_i$ cancels out the contribution to $\ss_i$ of all such edges. Thus, only edges that are cut in the residual graph contribute excess source. 

    Observe that the total volume of $\vol_G(S_i) \leq (64 \log_2 n) \cdot 100 \cdot \lambda \cdot 2^{k - i}/\phi$ because all the sinks inside $S_i$ are saturated by \Cref{fct:dinitz} and each vertex is a sink of at least $\frac{1}{\lambda}$ times its volume. We have that 
    \begin{align*}
        |\widehat{E}_{\ff}(S_i, \widehat{V} \setminus S_i)| &\leq \frac{\phi}{10^6  \cdot \lambda \cdot \log_2 n} \left(\vol_G(S_i) + \norm{\xx_{\ss_{i}, \tt_{i}, \ff_{i}}}_1\right) \\
        &< \frac{\phi}{10^6 \cdot \lambda \cdot \log_2 n}\left(2 \cdot (64 \log_2 n) 100 \cdot \lambda \cdot 2^{k - i}/\phi\right) \\
        &\leq 2^{k - i}/16
    \end{align*}
    where the first inequality follows from the termination condition of the while loop, the second inequality follows from the upper bound on the initial excess and total volume.

    Therefore, the total excess is at most $2^{k - i}/\phi$, since every such edge contributes at most $16/\phi$ total source (it contributes only to the endpoint not in $S_i$ with $8/\phi$ added to $\ss_i$ and at most $8/\phi$ from $\ddelta_i$). This concludes the proof. 
\end{proof}

Given that at every level $i$, the amount of new source added due to deleted edges in $B_i$ is at most $\frac{16}{\phi} 2^{k - i}$, we could directly conclude that the algorithm produces a valid flow whenever all the sets $A_i$ are empty. In the next section, we carefully describe the batching scheme and prove that it always produces a flow routing the demands. 

\subsection{A Dynamic Algorithm via Batching} \label{sec:batching}

Before we wrap up the analysis, we describe the update scheme that we employ. While simpler update schemes exist, this particular one is tailored towards de-amortization. This deamortization is explained in detail in \Cref{sec:worst_case_low_rec}.

\paragraph{The Update Scheme. } In this paragraph we describe how we run \Cref{alg:batch_pruning} in conjunction with a batching scheme. Whenever an new deletion occurs, we update the contents of the sets $B_i, \ldots, B_k$ and $A_i, \ldots, A_{\lambda}$ for some $i$. Then, we re-start the main-loop at \Cref{alg:batch_pruning:main_loop} of \Cref{alg:batch_pruning} at index $i$. The total number of updates the algorithm processes is $\phi 2^{k - 1} /(10^7 \log_2^6 n)$.

Concretely, all the deletions received so far are stored in the sets $B_1, \ldots, B_k$, where every edge deletion is stored in a distinct set $B_i$. For all $i$, the set $B_i$ always contains either $0, 2^{k - i - 1}$ or $2^{k - i}$ edges. 

\begin{definition}
    If the set $B_i$ contains $0$ edges we call it empty, if it contains $2^{k - i - 1}$ edges we call it half-full, and if it contains $2^{k - i}$ edges we call it full. 
\end{definition}

Initially, all the sets $\{B_i\}_{i = 1}^{\lambda}$ and $\{A_i\}_{i = 1}^{\lambda}$ are empty. We then process updates as follows. To enable our update scheme, we additionally ensure that the set $B_k$ is always empty after we finished processing a deletion. 
\begin{itemize}
    \item When a new deletion arrives, it first gets added to $B_k$.
    \item We let $i \in [k]$ be the largest index such that $B_i$ is either empty or half-full. Notice that $i < k$ because $B_k$ is full when it contains a single edge. %
    We then say the update triggers a rebuild at level $i$ and we say that all layers $j \geq i$ are affected by this rebuild. 
    \item We update the level $i$ as follows: $B_i \gets B_i \cup B_{i + 1}$. $A_i \gets S_i \cup A_i \cup S_{i + 1} \cup A_{i + 1}$. 
    \item We update levels $j = i + 1, \ldots, \lambda$ as follows: $B_j \gets B_{j + 1}$, $A_j \gets S_{j + 1} \cup A_{j + 1}$. To simplify the description, we let $B_{\lambda + 1}, A_{\lambda + 1}, S_{\lambda + 1} = \emptyset$ be empty sets. Notice that $B_{k}, \ldots, B_{\lambda}$ are then empty by the description of our algorithm. 
    \item Then, we re-run the main-loop  at \Cref{alg:batch_pruning:main_loop} of \Cref{alg:batch_pruning} from index $i$. 
\end{itemize}

We notice that this description directly ensures that the graph $\widehat{G}$ considered at level $i$ of the algorithm is a smaller and smaller sub-graph of $G$ as time goes on. This provides the monotonicity we are after. We first show that the volume of the pruned set grows nicely with the number of deletions. 

We then prove a simple claim about our batching scheme that is required for synchronization. 

\begin{claim}
    \label{clm:rebuild_timing}
    Whenever level $i$ is affected by a rebuild, it is not affected by a rebuild for the next $2^{k - i - 1} - 1$ deletions.
\end{claim}
\begin{proof}
    The batch $B_i$ for every level $i = 0, \ldots, k$ is initialized to be empty. Then, whenever an edge $e$ is deleted, we look for the empty or half-empty batch (containing $0$ or $2^{k - l - 1}$ edges) of highest index $l$. Thereafter, all the batches larger than this batch are half-full, and $i$ can only be affected by a rebuild after they are all full again. This can only happen after $2^{k - i - 1} - 1$ additional deletions. 
\end{proof}

\begin{claim}[Volume of the Pruned Set] \label{clm:volume}
    After $D$ deletions occurred, we have $\vol_G(\bigcup_{j \in [\lambda]} S_j \cup A_j) \leq O(D \cdot \frac{\log^3 n}{\phi})$. 
\end{claim}
\begin{proof}
    Whenever a layer $i$ is affected by a rebuild it adds at most volume $\vol_G(S_i) \leq (64 \log_2 n) 100 \cdot \lambda \cdot 2^{k - i}/\phi$ to the pruned set by \Cref{clm:source_reduction}. Since layer $i$ is affected by a rebuild after $2^{k - i - 1}$ deletions since the last rebuild and $\lambda \leq 8 \log_2 n$, the claim follows. 
\end{proof}

However, these sets $A_i$ unfortunately also complicate the analysis, because we have to carefully bound the amount of extra excess they introduce. By examining our algorithm, we observe that every set $A_i$ is the union of various (previous) sets $S_j$ where $j \geq i$. When these were initially pruned, they were chosen because they did not increase the source by too much. In the following, we show that this is a structural property independent of the current flow $\ff_i$. We first define bottleneck cuts, the structural property we aim to exploit. See also \Cref{fig:bottleneck} for an illustration of \Cref{def:bottleneck_cut}. 

\begin{definition}[$\gamma$-Bottleneck Cut] \label{def:bottleneck_cut}
    Let $G = (V, E)$, $A \subseteq V$, and $B \subseteq E$. We call $C$ a $\gamma$-bottleneck cut of $G$ at level $i$ with respect to $(A, B)$ if for all $A' \supseteq A$ and $B' \supseteq B$ we have 
    \begin{align*}
        \frac{8}{\phi}(|\partial_G(A' \cup C)| - |\partial_G(A')| - (\deg_{G[V \setminus A']} - \deg_{G[V \setminus A'] \setminus B'})(C \setminus A')) \leq \gamma - \frac{i}{\lambda} \vol_G(C \setminus A').
    \end{align*}
\end{definition}

\begin{figure}
    \centering
    \includegraphics[width=0.35\linewidth]{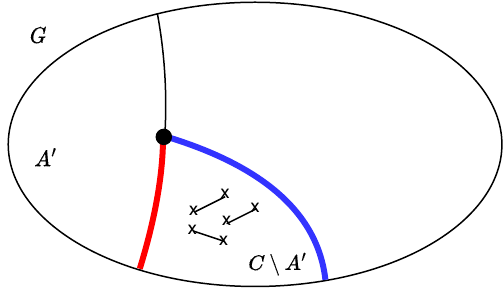}
    \caption{$C$ is a $\gamma$-bottleneck cut with respect to $(A, B)$ if $\frac{8}{\phi}(\text{\#red} - \text{\#blue}) + \frac{8}{\phi}(\deg_{G[V \setminus A']} - \deg_{G[V \setminus A'] \setminus B})(C \setminus A') + \frac{i}{\lambda}\vol_G(C \setminus A') \leq \gamma$ for every set $A' \supseteq A$.  Notice that $\frac{i}{\lambda}\vol_G(C \setminus A') = \tt(C \setminus A')$ and that $\frac{8}{\phi}(\deg_{G[V \setminus A']} - \deg_{G[V \setminus A'] \setminus B})(C \setminus A')$ corresponds to the amount of source added for endpoints of deleted edges in $B$ inside the set $C \setminus A'$. In the figure, the deleted edges are marked with $x$ at both endpoints. For each such endpoint, a deleted edge adds $\frac{8}{\phi}$ source.}
    \label{fig:bottleneck}
\end{figure}
We first observe that we can can add to the sets $A$ and $B$ and decrease the level while maintaining that a cut is a bottleneck cut. 

\begin{observation}[Relaxations] \label{clm:weaker}
    Let $G = (V, E)$, $A \subseteq V$, and $B \subseteq E$. Let $C$ a $\gamma$-bottleneck cut of $G$ at level $i$ with respect to $(A, B)$. Then, C is a $\gamma$-bootleneck cut of $G$ at level $j \leq i$ with respect to $(A', B')$ for all $A' \supseteq A$ and $B' \supseteq B$. 
\end{observation}
\begin{proof}
    Directly follows from the definition of bottleneck cuts. 
\end{proof}

Next, we show that bottleneck cuts compose nicely. 

\begin{claim}[Composition] \label{clm:comp}
    Assume that $C_1$ is a $\gamma_1$-bottleneck cut of a graph $G$ with respect to $(A, B)$ at level $i$, and that $C_2$ is a $\gamma_2$-bottleneck cut of a graph $G$ with respect to $(A \cup C_1, B)$ at level $i$. Then $C_1 \cup C_2$ is a $(\gamma_1 + \gamma_2)$-bottleneck cut with respect to $(A, B)$ at level $i$. 
\end{claim}
\begin{proof}
    For $A' \supseteq A$, we have 
    \begin{align} \label{eq:decomp_boundary}
        |\partial_G(A' \cup C_1 \cup C_2)|  - |\partial_G(A')| = |\partial_G(A' \cup C_1 \cup C_2)| - |\partial_G(A' \cup C_1)| + |\partial_G(A' \cup C_1)| - |\partial_G(A')|.
    \end{align}
    We also have 
    \begin{align} \label{eq:decomp_bsource}
        - (\deg_{G[V \setminus A']} - \deg_{G[V \setminus A'] \setminus B'})((C_1 \cup C_2) \setminus A')) = &- (\deg_{G[V \setminus A']} - \deg_{G[V \setminus A'] \setminus B'})(C_1 \setminus A')) \\ & - (\deg_{G[V \setminus A']} - \deg_{G[V \setminus A'] \setminus B'})(C_2 \setminus (A'\cup C_1))). 
    \end{align}
    To bound 
    \begin{align*}
        |\partial_G(A' \cup C_1 \cup C_2)|  - |\partial_G(A')| - (\deg_{G[V \setminus A']} - \deg_{G[V \setminus A'] \setminus B'})((C_1 \cup C_2) \setminus A'))
    \end{align*}
    it then suffices to bound 
    \begin{align} \label{eq:first_prune}
        |\partial_G(A' \cup C_1 \cup C_2)| - |\partial_G(A' \cup C_1)| - (\deg_{G[V \setminus A']} - \deg_{G[V \setminus A'] \setminus B'})(C_2 \setminus (A'\cup C_1)))
    \end{align}
    and
    \begin{align} \label{eq:second_prune}
        |\partial_G(A' \cup C_1)| - |\partial_G(A')| - (\deg_{G[V \setminus A']} - \deg_{G[V \setminus A'] \setminus B'})(C_1 \setminus (A' \cup C_1)) 
    \end{align}
    separately by \eqref{eq:decomp_boundary} and \eqref{eq:decomp_bsource}. By our assumption \eqref{eq:first_prune} is bounded by $\gamma_2 - \frac{i}{k}\vol_G(C_2 \setminus A')$ and \eqref{eq:second_prune} is bounded by $\gamma_1 - \frac{i}{\lambda}\vol_G(C_1 \setminus A')$. We then conclude that $\eqref{eq:first_prune} + \eqref{eq:second_prune}$ is bounded by $\gamma_1 + \gamma_2 - \frac{i}{\lambda}\vol_G(C_1 \cup C_2 \setminus A')$. This concludes the proof of the claim. 
\end{proof}

In the following, we show that pruning a $\gamma$-bottleneck cut does not increase excess by more than $\gamma$. This will be sufficient for bounding the extra pruning due to the sets $A_i$. 

\begin{claim}[Excess Preservation] \label{clm:excess_preservation}
    Consider the $i$-th iteration of the \Cref{alg:batch_pruning}. Assume that after replacing $A_i$ with $\emptyset$, the excess after \Cref{alg:batch_pruning:graph} of \Cref{alg:batch_pruning} is at most $\xi$. Furthermore, assume that $A_i$ is a $\gamma$-bottleneck cut with respect to $\left(\bigcup_{j < i} A_j \cup S_j, \bigcup_{j \leq i} B_j\right)$. Then, the excess after \Cref{alg:batch_pruning:graph} of \Cref{alg:batch_pruning} of batch pruning is at most $\xi + \gamma$.
\end{claim}
\begin{proof}
Let $\hat{V}_i^{\text{(before)}} = V \setminus \left(\bigcup_{j < i} A_j \cup S_j\right)$ denote the vertex set of the graph at the start of the $i$-th iteration, $A_i' = A_i \cap V_i^{\text{(before)}}$ and $\hat{V}_i = \hat{V}_i^{\text{(before)}} \setminus A_i$. 

We start by considering the amount of excess that is currently placed on $A_i'$ (and hence removed when we take away $A_i$ from the graph). Let $\eta$ denote the total amount of flow routed across the cut from $A'_i$ to $\hat{V}_i^{\text{(before)}}$, i.e. $\eta = \sum_{(u,v) \in \partial_{G[\hat{V}_i^{\text{(before)}}]}(A'_i)} \ff_{i - 1}(u, v)$. Then, the excess inside $A'_i$ is exactly 
\[
\ss_i(A_i')  - \frac{i}{\lambda} \vol_G(A_i') - \eta
\]
since $\ff_{i-1}$ is supported on $\hat{V}^{(before)}$ and thus the excess on $A_i'$ is the amount of flow that originates from $A_i'$, i.e. $\ss_i(A_i')$, minus the amount of flow absorbed on $A_i'$ or routed away.

On the other hand, the amount of new excess that is produced by inducing on the complement of $A_i'$ can be upper bounded by 
\[
\frac{8}{\phi}|\partial_{G[\hat{V}_i^{\text{(before)}}]}(A'_i)| - \eta
\]
because every edge in the cut can contribute at most $8/\phi$ units of source and inducing the flow means that the flow previously routed into $A_i'$ is now no longer routed and thus considered excess.

Thus, the net change in excess can be upper bounded by
\begin{align*}
    \xi - \underbrace{\left(\ss_i(A'_i) - \eta - \frac{i}{\lambda}\vol_G(A'_i)\right)}_{\geq \text{ excess inside $A'_i$}} &+  \underbrace{\frac{8}{\phi}|\partial_{G[\hat{V}_i^{\text{(before)}}]}(A'_i)| - \eta}_{= \text{ additional excess}} \\ &= \xi + \frac{8}{\phi}|\partial_{G[\hat{V}_i^{\text{(before)}}]}(A'_i) - \ss_i(A'_i) - \frac{i}{\lambda}\vol_G(A'_i).
\end{align*}
It remains to observe that for $C = A_i'$ and $A' = V \setminus \hat{V}_i^{\text{(before)}} $, we have 
    \begin{enumerate}
        \item $|\partial_{G[\hat{V}_i^{\text{(before)}}]}(A'_i)| - |\partial_{G[A_i]}(A'_i)| = |\partial_G(A' \cup C)| - |\partial_G(A')|$: this follows by investigating the cuts carefully.
        \item $\ss_{i}(A'_i) \geq \frac{8}{\phi}((\deg_{G[V \setminus A']} - \deg_{G[V \setminus A'] \setminus B'})(C \setminus A') + |\partial_{G[A_i]}(A'_i)|)$: it can be seen from the algorithm that $\ss_{i}(A'_i) = \frac{8}{\phi}((\deg_{G} - \deg_{G[V \setminus A'] \setminus B'})(C \setminus A')$ and $\deg_G(C \setminus A')$ can be seen to be the sum of degree within the subgraph induced on $G \setminus A'$ and all the edges incident to $C \setminus A'$ leaving to $A'$. 
    \end{enumerate}
It follows that the total excess can be upper bounded by 
\begin{align*}
\xi &+ \frac{8}{\phi}|\partial_{G[\hat{V}_i^{\text{(before)}}]}(A'_i) - \ss_i(A'_i) - \frac{i}{\lambda}\vol_G(A'_i)\\
&\leq \xi + |\partial_G(A' \cup C)| - |\partial_G(A')| - (\frac{8}{\phi}((\deg_{G[V \setminus A']} - \deg_{G[V \setminus A'] \setminus B'})(C \setminus A')) - \frac{i}{\lambda}\vol_G(A'_i)\\
&\leq \xi + \gamma
\end{align*}
where the last inequality stems from the definition of bottleneck cuts in \Cref{def:bottleneck_cut} and the assumptions of the claim.  
\end{proof}

Given \Cref{clm:excess_preservation}, the main remaining difficulty lies in proving that all $A_i$ are $\gamma$-bottleneck cuts with respect to $\left(\bigcup_{j < i} A_j \cup S_j, \bigcup_{j < i} B_j\right)$ at level $i$. 

\begin{claim}
    Assume that the total excess after \Cref{alg:batch_pruning:graph} of \Cref{alg:batch_pruning} in iteration $i$ is at most $\norm{\xx_{\ss_{i}, \tt_{i}, \ff_{i}}}_1 \leq 100 \cdot 2^{k - i}$, then we have that $S_i$ is a $\gamma_i$-bottleneck cut at level $i$ with respect to $\left(A_i \cup \bigcup_{j < i} A_j \cup S_j, \bigcup_{j \leq i} B_i\right)$ where $\gamma_i = 2^{k - i}/\phi$. 
\end{claim}
\begin{proof}
Assume that $S_i \neq \emptyset$ is not empty. Otherwise, we are done. 
    
The while loop terminated by \Cref{clm:while_terminates}. Therefore, the flow $\ff_i$ routed all but $2^{k - i}/\phi$ of the source inside $S_i$. But now consider a larger pruned set or more deletions. Every extra deletion can cause at most $8/\phi$ flow to no-longer be routed out, but also introduces at least $8/\phi$ extra source flow inside $S_i$. The claim follows. 
\end{proof}

\begin{invariant} \label{lem:main_invariant}
    The following invariants hold after processing any update.  
    \begin{enumerate}
        \item For all $i \in [\lambda]$, $A_i$ is a $12 \gamma_i \log_2 n$-bottleneck cut with respect to $\left(\bigcup_{j < i} A_j \cup S_j, \bigcup_{j < i} B_j\right)$ at level $i$ where $\gamma_{i} = 2^{k - 1}/\phi$. \label{inv:item1}
        \item The total excess after \Cref{alg:batch_pruning:graph} of \Cref{alg:batch_pruning} is at most $\norm{\xx_{\ss_{i}, \tt_{i}, \ff_{i - 1}}}_1 \leq (64 \log_2 n) 100 \cdot 2^{k - i} / \phi$, $A_i = \emptyset$, and that $\norm{\ss_{i}}_1 \geq \norm{\tt_{i}}_1$. \label{inv:item2}
        \item For all $i \in [\lambda]$, $\vol_G(A_i \cup S_i) \leq 10^7 \log_2^3(n) \cdot 2^{k - i}/\phi$. 
        \label{inv:item3}
    \end{enumerate}
\end{invariant}
\begin{proof}
    Before we prove the invariant, we strengthen \Cref{inv:item1} to simplify the proof. In the following, we prove that:
    \begin{itemize}
        \item If $B_i$ is full, then $A_i$ is a $4(\lambda - i + 1) \gamma_i$-bottleneck cut at level $i$. 
        \item If $B_i$ is half-full or empty, then $A_i$ is a $2(\lambda - i + 1) \gamma_i$-bottleneck cut at level $i$.
    \end{itemize}
    We first observe that whenever $B_i$ is empty, so is $A_i$ by the description of our batching scheme because then all sets $B_{i'}$ (and $A_{i'}$ by induction) with $i' \leq i$ are empty as well. Therefore the invariant follows directly. Furthermore, we observe that whenever layer $i$ is affected by an update and $B_i$ was full before the update, $B_i$ is empty after the update. Thus, \Cref{inv:item1} is satsified after initialization. 
    \begin{itemize}
        \item \underline{\Cref{inv:item1} implies \Cref{inv:item2}:} We first show that if \Cref{inv:item1} holds after an update, then \Cref{inv:item2} follows by induction. For level $1$, we have that the initial amount of source is at most $\frac{8}{\phi} 2^{k - 1}$, and therefore the invariant is true at level $1$ by \Cref{clm:excess_preservation}. Then, at layer $i$, we have that the excess from the previous layer is at most by $2^{k - i + 1}/\phi$. The new excess due to $B_i$ is at most $\frac{8}{\phi} 2^{k - i}$. Therefore, the invariant at level $i$ follows from \Cref{clm:excess_preservation} and \Cref{inv:item2} since $\lambda < 8 \log_2 n$
        \item \underline{\Cref{inv:item2} implies \Cref{inv:item1} after update:} We now show  that \Cref{inv:item2} holds before an update implies that \Cref{inv:item1} holds after updating the sets $A_i$. This then, implies that \Cref{inv:item2} holds after re-running the algorithm from index $i$ by the previous item. 

        We refer to some object $X$ before and after the update as $X^{\text{(before)}}$ and $X^{\text{(after)}}$ respectively. Then, we have $A_i^{\text{(after)}} = A_i^{\text{(before)}}  \cup S_i^{\text{(before)}} \cup A_{i + 1}^{\text{(before)}} \cup S_{i + 1}^{\text{(before)}}$. Firstly, $A_i^{\text{(before)}} \cup S_{i + 1}^{\text{(before)}}$ is a 
        \begin{align*}
            2(\lambda - i + 1)\gamma_i + \gamma_i 
        \end{align*} bottleneck cut at level $i$ with respect to $(\bigcup_{j < i} A_j \cup S_j, \bigcup_{j \leq i} B_j)$ by the invariant, \Cref{clm:weaker}, and \Cref{clm:comp}. Then, $A_i^{\text{(before)}}  \cup S_i^{\text{(before)}} \cup A_{i + 1}^{\text{(before)}}$ is a 
        \begin{align*}
            4(\lambda - i + 1)\gamma_i - \gamma_i 
        \end{align*} bottleneck cut at level $i$ with respect to $(\bigcup_{j < i} A_j \cup S_j, \bigcup_{j \leq i} B_j)$, again by the invariant, \Cref{clm:weaker}, and \Cref{clm:comp}. Finally, another application of these claims yields that 
        $A_i^{\text{(before)}}$ is a $4(\lambda - i + 1)\gamma_i$-bottleneck cut as desired. 

        Finally, for the levels $i' > i$, we have that $A_{i'}^{\text{(after)}} = A_{i' + 1}^{\text{(before)}} + S_{i' + 1}^{\text{(before)}}$ is a $2(\lambda - i' + 1)\gamma_{i'}$-bottleneck cut  for $(\bigcup_{j < i'} A_j \cup S_j, \bigcup_{j \leq i} B_j)$ at level $i'$, again by the invariant, \Cref{clm:weaker}, and \Cref{clm:comp}.
        \item \underline{Proof of \Cref{inv:item3}:} Given \Cref{inv:item1} and \Cref{inv:item2}, we now prove \Cref{inv:item3}. We have $\vol_G(S_i) \leq (64 \log_2 n) 100 \cdot \lambda \cdot 2^{k - i}/\phi \leq 10^5 \log^2(n) \cdot 2^{k - i}/\phi$ throughout by \Cref{clm:source_reduction} and $\lambda \leq 8 \log_2 n$. We then strengthen the invariant to $\vol_G(A_i) \leq 2(\lambda - i + 1) \leq 10^5 \log^2(n) \cdot 2^{k - i}/\phi$ when $B_i$ is half-full or empty and $\vol_G(A_i) \leq 4(\lambda - i + 1) \leq 10^5 \log^2(n) \cdot 2^{k - i}/\phi$ when $B_i$ is full. The invariant intially holds because all sets are empty. Then, whenever a rebuild at layer $i$ happens, it is full after and the volume is bounded by the sum of 
        \begin{align*}
            \vol_G(A_i^{\text{(before)}}) + \vol_G(A_{i + 1}^{\text{(before)}}) \leq 4(\lambda - i + 1) \leq 10^5 \log^2(n) \cdot 2^{k - i}/\phi - 2 \cdot 10^5 \log^2(n) \cdot 2^{k - i}/\phi
        \end{align*}
        and
        \begin{align*}
            vol_G(S_i^{\text{(before)}}) + \vol_G(S_{i + 1}^{\text{(before)}}) \leq 2 \cdot 10^5 \log^2(n) \cdot 2^{k - i}/\phi. 
        \end{align*}
        For the layers $j > i$, we have that
        \begin{align*}
            \vol_G(A_{i + 1}^{\text{(before)}}) + \vol_G(S_{i + 1}^{\text{(before)}}) \leq 4(\lambda - j + 1) \leq 10^5 \log^2(n) \cdot 2^{k - i}/\phi
        \end{align*}
        as desired. The item follows directly from the definition of the new sets $A_j^{(after)}$. 
    \end{itemize}   
    Finally, we remark that $\norm{\ss_{i}}_1 \geq \norm{\tt_{i}}_1$ throughout because the total amount of source at layer $i$ is at most $(24 \log_2 n) 100 \cdot 2^{k - i} / \phi$. But because our pruning algorithm ensures stops whenever $2^{k - 1}/(10^7 \log_2^4 n)$ edges are deleted, and a layer only contributes source if it is half-full or full, this is at most $2^{k - 1}/(10 \phi \log_2^3 n)$. Therefore, this condition is fulfilled throughout. 
\end{proof}

Next, we show that the amortized runtime is $\tilde{O}(\phi^{-1})$

\begin{claim} \label{clm:runtime}
    The algorithm processes $D$ deletions in total time $\tilde{O}(D/\phi)$ and rebuilding layer $i$ takes time $\tilde{O}(2^{k - i}/\phi)$.
\end{claim}
\begin{proof}
    Layer $i$ gets rebuilt every $2^{k - i - 1}$ deletions, and a rebuild of layer $i$ takes time $\tilde{O}(2^{k - i}/\phi^2)$ since the runtime is dominated by the call to \textsc{Dinitz}(), whose runtime follows from \Cref{lem:main_invariant} and \Cref{fct:dinitz}. The claim follows from summing up the layers.
\end{proof}

We finally prove that the graph $G[\hat{V}] \setminus B$ remains a $\Omega(\phi)$-expander. 

\begin{lemma}\label{lem:maintains_exp}
    $G[\hat{V}] \setminus B$ is a $\phi/10$-expander throughout.
\end{lemma}
\begin{proof}
    The batching scheme maintains a valid flow certificate since the excess of the final layer is $0$. Therefore, the lemma follows from \Cref{lem:flow_cert_guarantees_exp}. 
\end{proof}

Given the previous claims, we could already give an expander pruning algorithm that processes every update in poly-logarithmic time by noting that the rebuild at layer $i$ can be started as soon as layer $i + 1$ is full, which allows us to adequately distribute the computational cost. However, this algorithm would still require a pointer switch when the pre-processing finished. This means that the set $\vol_G(\bigcup_{j \in [\lambda]} S_j \cup A_j)$ sometimes grows a lot, and an algorithm could not read all the updates in poly-logarthmic time after an update finished. We will address this caveat in the following two sections. 

\section{Low Worst-Case Recourse Expander Pruning}
\label{sec:amortized_low_rec}

The goal of this section is to prove the following theorem that formalizes guarantees for our low worst-case recourse expander pruning algorithm, albeit still with amortized runtime.

\begin{restatable}{theorem}{lowRecourseAmortTime}
\label{thm:amortized_low_recourse}
    There exists an algorithm that given a $\phi$-expander $G$ and a sequence of up to $\tilde{\Omega}(\phi m)$ deletions to $G$ maintains a set $S_0 \subseteq V$ such that $S_0$ grows by at most $\tilde{O}(1/\phi^2)$ vertices after every deletion and $G[V \setminus S_0]$ remains a $\Omega(\phi/\log_2^4 m)$-expander throughout.

    The total processing time of $D$ deletions is $\tilde{O}(D/\phi^2)$.
\end{restatable}

To do so, we introduce flow certificates and describe their maintenance under edge deletions. These are crucial objects for our algorithm allowing us to slowly prune a discarded part while ensuring that the original graph remains a good expander throughout. On a high level, we run the dynamic batch pruning algorithm $\textsc{BatchPrune}()$ (\Cref{alg:batch_pruning}) with the dynamic batch update scheme introduced in \Cref{sec:batching} while slowly pruning the discarded parts $S_i$ instead of removing them all at once. If we ensure that whenever a rebuild affects layer $i$, the set $S_i$ is fully pruned, then this algorithm nicely interfaces with the batching scheme. 

\subsection{Flow Certificates and their Nesting}

We first define flow certificates. They certify that a discarded part can be kept around while approximately maintaining expansion.

\begin{definition}[Flow certificate]
    Given a graph $G = (V, E)$, we say a tuple $(\ff, S)$ is a $(\gamma_{\text{source}}, \gamma_{\text{sink}}, c)$-\emph{flow certificate} for $G$ where $\gamma_{\text{source}}, \gamma_{\text{sink}}, c \in \mathbb{R}^+$ with $c \geq 1$ if we have:
    \begin{enumerate}
        \item $S \subset V$ and
        \item  $\ff$ is a feasible flow on the current graph $G$ that routes at least $\gamma_{\text{source}} \cdot \deg_{G}(s)$ source from every $s \in S$ to sinks $t \in V \setminus S$ of capacity at most $\gamma_{\text{sink}} \cdot \deg_{G}(t)$ with capacity $c$ on every edge.
    \end{enumerate}
    \label{def:flow_cert}
\end{definition}

If $G[V \setminus S]$ has good expansion, then a flow certificate ensures that the whole graph $G$ is a good expander. 

\begin{lemma}
    \label{lem:flow_cert_guarantees_exp}
    Let $G = (V, E)$ be a graph and let $S \subseteq V$ be such that $G[V \setminus S]$ is an $\alpha$-expander. Then, a $(\gamma_{\text{source}}, \gamma_{\text{sink}}, c)$-flow certificate $(S, \ff)$ where $\gamma_{\text{source}} \geq  \gamma_{\text{sink}}/\delta \geq 1$ for $\delta \geq 1$ and $c \geq 3/\alpha$ shows that $G$ is a $\frac{1}{3(c + \delta)}$-expander.
\end{lemma}
\begin{proof}
Consider any cut $C \subset V$ with $\vol_{G}(C) \leq \vol_{G}(V \setminus C)$. Let $A = C \cap S$ and $B = C \setminus S$. Let us prove the claim by case distinction.

\underline{If $\vol_{G}(A) \geq (1 - \frac{1}{3\delta}) \vol_{G}(C)$:} This implies $\vol_{G}(B) < \frac{1}{2\delta} \vol_{G}(C)$ and $\vol_{G}(A) \geq \frac{2}{3}$ since $\delta \geq 1$. Thus at least $\gamma_{\text{source}} \cdot \vol_{G^{(0}}(A) - \gamma_{\text{sink}} \cdot \vol_{G}(B) \geq \gamma_{\text{source}} \cdot (\vol_{G^{(0}}(A) - \delta \vol_{G}(B)) \geq \frac{\gamma_{\text{source}}}{3} \cdot \vol_{G}(C)$ source mass on $C$ has to be routed out of $C$. But since $(\ff, S)$ is a flow certificate, we have that $\ff$ sends this amount of flow over the cut $E_{G}(C, V \setminus C)$ and since each edge can transport at most $c$ units of flow, we have that $|E_{G}(C, V \setminus C)| \geq \frac{\gamma_{\text{source}}}{3c} \cdot \vol_{G}(C) \geq \frac{1}{3c} \cdot \vol_{G}(C)$, as desired since $\delta \geq 1$.

\underline{Otherwise, $\vol_{G}(A) < (1 - \frac{1}{3\delta}) \vol_{G}(C)$:} this implies $\vol_{G}(B) \geq \frac{1}{3\delta} \vol_{G}(C)$. But from \Cref{def:pruned_exp}, we have that $|E_{G \setminus S}(B, V \setminus (B \cup S))| \geq \vol_{G}(B) \geq \frac{1}{\delta 3} \vol_{G}(C)$ and since $E_{G \setminus S}(B, V \setminus (B \cup S)) \subseteq E_G(C, V \setminus C)$, the claim follows.
\end{proof}

We will make extensive use of the following simple observation that enables us to compose flow certificates by simply adding them up.  

\begin{observation}
    \label{fct:flow_cert_comp}
    Given a $(\gamma_{\text{source}}, \gamma_{\text{sink}}, c)$ flow certificate $(S, \ff)$ for a graph $G$ and a $(\gamma_{\text{source}} + \gamma_{\text{sink}}, \gamma_{\text{sink}}', c')$ flow certificate $(S', \ff')$ for $G \setminus S$, we have that $(S \cup S', \ff + \ff')$ is a $(\gamma_{\text{source}}, \gamma_{\text{sink}} + \gamma_{\text{sink}}', c + c')$ flow certificate. 
\end{observation}

\subsection{Batch Flow Certificate Maintenance}

In this section, we describe our data structure for maintaining flow certificates (\Cref{def:flow_cert}) through a sequence of edge deletions/flow removals. 

\paragraph{Dynamic Trees.} We use the following classic link-cut tree data structure from \cite{ST83} that allows us to update flows on (directed) trees. 

\begin{definition}
We say a directed graph $\dir{F} = (\dir{V}, \dir{E})$ is a rooted forest if each connected component $T$ in $\dir{G}$ (where connected is in the undirected sense) consists of vertices with exactly one in-edge except for a single vertex per component, the root vertex, that has no in-edge.
\end{definition}
\begin{lemma}[Dynamic trees, see \cite{ST83}]
\label{algo:LCT}
Given a static, $m$-edge, directed graph $\dir{G} = (\dir{V}, \dir{E})$, there is a deterministic data structure $\textsc{DynTree}()$ that maintains a rooted forest $T \subseteq \dir{G} = (\dir{V}, \dir{E})$ under insertion/deletion of edges and maintains a flow $\ff \in \mathbb{R}^{\dir{E}}$:
\begin{enumerate}
    \item $\textsc{Initialize}(\ff_{init})$: Initializes the collection of trees to be the empty graph with no edges. And initializes the flow $\ff$ to $\ff_{init}$ via a pointer change.
    \item $\textsc{Insert}(e)$ / $\textsc{Delete}(e)$: Insert/delete edges $e$ to/from $T$, under the condition that $T$ is always a rooted forest, i.e. each vertex has at most one in-coming edge.
    \item $\textsc{FindRoot}(u)$: given vertex $u$, returns the root of $u$ in the tree $T$, i.e. the first vertex on the maximal path containing $u$.
    \item $\textsc{FindMin}(u)$: Given vertex $u \in V$, returns the edge $e$ on the root-of-$u$ to $u$ path that carries minimal flow $\ff(e)$ (and if there are multiple such edges, it returns the edge $e$ closest to $u$). 
    \item $\textsc{UpdateFlow}(u, \Delta)$: For any real $\Delta \in \mathbb{R}$, and vertex $u \in V$, adds $\Delta$ to the flow of every edge $e$ on the root-of-u to u path in $T$.
    \item $\textsc{ReadCurrentFlow}(e)$: Returns current $\ff(e)$.
\end{enumerate}
The data structure requires time $\tilde{O}(1)$ for initialization. Each additional operation can be implemented in worst-case time $O(\log m)$.
\end{lemma}

\paragraph{Batch Flow Certificate. }

We first state the main theorem of this section. Notice that while we do not explicitly require an upper bound on the size of the discarded sets $S'_i$ in this section, we do need to ensure that whenever a set gets replaced the previous set has been fully pruned. When we use this data structure in our full algorithm, the size of the set $S'_i$ will be roughly proportional to $\approx 2^{k - i}$. The algorithm then issues additional calls to $\textsc{RemoveEdge}()$ targeting the volume of these sets to ensure that the sets steadily get pruned. 

\begin{theorem}[Batch Flow Certificate]
    \label{thm:batch_flow_cert}
    There is a data structure \textsc{BatchFlowCert} (\Cref{alg:flow_cert_maint}) that supports the following operations: 
    \begin{enumerate}
        \item $\textsc{Initialize}(G, \phi, k)$: Initializes the graph $\widehat{G}$ maintained by the data structure with $G$. Sets $S_0, \ldots S_k \gets \emptyset$. These sets remain disjoint. We let $S := \bigcup_{i = 0}^k S_i$. The algorithm ensures that $S$ is monotonically growing.
        \item $\textsc{ReInitialize}(G, \{S'_l, \ldots S'_k \} , \phi)$: Assumes that 
        \begin{itemize}
            \item $l \geq 1$, 
            \item $S_i \subseteq S_0$ for all $i = l, \ldots, k$, 
            \item $\widehat{G}\left[V \setminus \left(\bigcup_{i = 0}^{l - 1} S_i \cup \bigcup_{i = l}^{j} S'_i \right)\right]$ is a $\phi/10$-expander for all $j = l, \ldots, k$, 
            \item $\vol_{G}(S_j) \leq \vol_{G}\left(V \setminus \left(\bigcup_{i = 0}^{l - 1} S_i \cup \bigcup_{i = l}^{j} S'_i\right)\right)/20 k$ for all $j = l, \ldots, k$. 
        \end{itemize}
        Sets $S_{i} \gets S'_i$ for all $i = l, \ldots, k$. 
        \item $\textsc{RemoveEdge}(e)$: Assumes that after removing edge $e$ from the graph $\widehat{G}$, we have that $\widehat{G}[V \setminus S]$ is a $\phi/10$-expander. Removes edge $e$ from $\widehat{G}$, adds up to $8000k^4/\phi$ vertices from the sets $\{S_1, \ldots, S_k\}$ to $S_0$.
    \end{enumerate}
    Then, the data structure maintains that the graph $\widehat{G}[V \setminus S_0]$ remains a $\Omega\left(\frac{\phi}{k^4}\right)$-decremental expander. The total runtime is $\tilde{O}(\vol_G(S)/\phi + k^4 D/\phi)$ where $D$ is the total number of calls to $\textsc{RemoveEdge}(e)$. Additionally, the operation $\textsc{RemoveEdge}(e)$ has worst-case runtime $\tilde{O}(k^4/\phi)$.
    \label{thm:flow_cert_maint} 
\end{theorem}

\paragraph{Data structure description. } 

We give an overview of the data structure \textsc{BatchFlowCert}() and describe its routines. We refer the reader to the pseudocode for a detailed description of the data structure (\Cref{alg:flow_cert_maint}).  

The data structure maintains a collection of sets $S_0, \ldots, S_k$, subsets $\widehat{S}_1, \ldots, \widehat{S}_k$ with $\widehat{S}_i \subseteq S_i$ that are initially empty, and a decremental graph $\widehat{G}$. The sets $S_l, \ldots, S_k$ can be updated with the routine $\textsc{ReInitialize}(\{S'_l, \ldots S'_k\})$ if the sets $\widehat{S}_i = \emptyset$ for $i = l, \ldots, k$ are currently empty. At that point, $S_i \gets S'_i$ and $\widehat{S}_i \gets S'_i \setminus S_0$ for all $i = l, \ldots, k$.

\begin{algorithm}
\caption{\textsc{BatchFlowCert}}
\label{alg:flow_cert_maint}
\SetKwFunction{algo}{\textsc{ReInitialize}}\SetKwFunction{proc}{\textsc{RemoveFlow}}
\SetKwFunction{procinit}{\textsc{Initialize}}
\SetKwFunction{procdel}{\textsc{RemoveEdge}}

\SetKwProg{myproc}{Procedure}{}{}

\myproc{\procinit{$G = (V,E), \phi, k$}}{
    \For{$i = 1, \ldots k$}{
        $S_i \gets \emptyset$;
        $\ff_i \gets \veczero$;
    }
    $\widehat{G} = (\widehat{V}, \widehat{E}) \gets G$; $S_0 \gets \emptyset$; $\dd \gets \deg_G$ \\
}

\myproc{\algo{$\{S'_l, \ldots S'_k \}$}}{
    \label{alg:flow_cert_maint:init}
    \For{$i = l, \ldots k$}{
        $S_i \gets S'_i$; $\widehat{S}_i \gets S'_i \setminus S_0$\\
        $\ff_i \leftarrow \textsc{Dinitz}(\widehat{G}\left[V \setminus \bigcup_{j = 0}^{i - 1} S_j \right], \frac{8000 k^3}{\phi}), \ss_i = (10i \cdot k + i + 1) \cdot \dd[S_i \setminus S_0], \tt_i = (10 \cdot k - 1) \cdot \dd \left[V \setminus \bigcup_{j = 0}^{i} S_j  \right], h = \frac{200}{\phi} \log_2 n)$ \label{alg:flow_cert_maint:dinitz} \tcp{Compute Maximum Flow on Expander} 
        Compute a path-cycle decomposition of $\ff_i$, and update $\ff_i$ by removing all cycles. \\
        $\dir{G}_i = (\dir{V}_i = V \setminus \bigcup_{j = 0}^{i - 1} S_j, \dir{E}_i = \supp(\ff_i))$ \\
    $\mathcal{D}_i \gets \textsc{DynTree}.\textsc{Initialize}(\ff_i)$; $\qq_j \gets \veczero_{S_j}$ \tcp*{Initialize empty tree.}  
        
    }
}
\myproc{\proc{$e = (u,v), i$}}{
    \label{alg:flow_cert_maint:del}
    \If{$(u,v) \in \dir{E}_i$ or $(v, u) \in \dir{E}_i$}{
        Assume wlog that $e = (u,v) \in \dir{E}_i$. \\
            \If{$\mathcal{D}_i.\textsc{ReadCurrentFlow}(e) > 0$}{
                \If{$\mathcal{D}_i.\textsc{ReadCurrentFlow}(\mathcal{D}_i.\textsc{FindMin}(v)) = 0$}{
                $e' = (x,y) \gets \mathcal{D}_i.\textsc{FindMin}(v)$ \\
                $\mathcal{D}_i.\textsc{Delete}(e')$;
                $\dir{E}_i \gets \dir{E}_i \setminus \{e'\}$
                }
                $\mathcal{D}_i.\textsc{UpdateFlow}(v, - 1)$\\
                $r \gets \mathcal{D}_i.\textsc{FindRoot}(v)$\label{alg:flow_cert_maint:root} \\ 
                \If{There exists $(x, r)$ in $\dir{E}_i$}{ \label{alg:flow_cert_maint:link}
                    $\mathcal{D}_i.\textsc{Insert}(x,r)$ \\
                    \If{$r \in S_i$}{$\qq_i(r) \gets \qq_i(r) + 1$}
                }\ElseIf{$r \in S_i$}{
                {
                    \If{there exists $(x, r)$ in $\dir{E}_j$ for some $j < i$}{
                       \If{$\mathcal{D}_j.\textsc{ReadCurrentFlow}(x,r) = 0$}{
                            \tcp{This edge actually carries no flow}
                            $\dir{E}_j \gets \dir{E}_j \setminus (x, r)$; 
                            $\qq_i(r) \gets \qq_i(r) + 1$ \\

                       }\Else{
                            \proc{(x, r), j}
                       }
                    }\Else{
                        $\qq_i(r) \gets \qq_i(r) + 1$ \\
                    }
                    \If{$\qq_i(r) > (i + 1) \dd(r)$}{
                        $\widehat{S}_j \gets \widehat{S}_j 
                        \setminus \{r\}$ \label{alg:flow_cert_maint:pruned};
                        $S_0 \gets S_0 \cup \{ r \}$ 
                    }
                }
                }
            }

    }
}
\myproc{\procdel{$e$}}{
        \For{$i = 1, \ldots, k$ and $j = 1, \ldots, 8000 k^3/\phi$}{
            $\proc{e, i}$
        }
    $\widehat{G} \gets \widehat{G} \setminus e$ \\
}
\end{algorithm}

The only other updates to these sets are by the algorithm which removes vertices from sets $\widehat{S}_i$ and adds them to $S_0$. Therefore, the set $S := \bigcup_{i = 0}^k S_i$ is monotonically growing throughout the execution of the algorithm. 

Assuming that $\widehat{G}[V \setminus S]$ remains a $\phi/10$ expander throughout, the data structure guarantees that $\widehat{G}[V \setminus S_0]$ remains a $\Omega(\phi/k^4)$ expander, and that $S_0$ only grows by at most $\tilde{O}(\phi^{-1})$ vertices after every edge deletion to $\widehat{G}$ (our full algorithm will issue $\tilde{O}(\phi^{-1})$ additional deletions per real edge deletion, which explains the extra $\phi^{-1}$ factor in the growth of the pruned set in \Cref{thm:amortized_low_recourse}, the main theorem of this section). To do so, it maintains a flow certificate (See \Cref{def:flow_cert}) $(\bigcup_{i = 1}^k \widehat{S}_i, \ff = \sum_{i = 1}^k \ff_i)$ composed of $k$ flow certificates $(\widehat{S}_i, \ff_i)$ each supported on a respective decremental graph $\widehat{G}_i = (\widehat{V}_i, \widehat{E}_i) := \widehat{G}\left[V \setminus \bigcup_{j = 0}^{i - 1} S_j\right]$.

When layer $i$ gets re-initialized, we compute a flow $\ff_i$ that routes $(10i \cdot k + i + 1)\dd[S_i]$ units of source flow to sinks of capacity $(10k - 1)\dd[\widehat{V}_i \setminus S_i]$ where $\dd := \deg_{G}$ is shorthand for the initial degree vector of the graph $G$. For computing each of these flows, we set the edge capacities to $8000k^3/\phi$.\footnote{The reader might notice that a lower source, sink and edge capacity would suffice for our amortized algorithm presented in this section. We already set these capacities higher by a poly-logarithmic factor $\poly(k) = \polylog m$ in anticipation of \Cref{sec:worst_case_low_rec}.} To facilitate backtracking, we assume that each flow $\ff_i $ is routed on a DAG and otherwise we simply remove cycles by computing a flow decomposition. Then, whenever an edge $e$ carrying flow $\ff_i$ is deleted from $\widehat{G}$, we repeatedly aim to backtrack the flow $\ff_i$ to a true source, i.e. a vertex that does not have any incoming flow and has a leaving flow path that sends flow over edge $e$. Our data structure will sometimes fail to backtrack to such a true source vertex, but every vertex $v$ is only erroneously backtracked to at most $\deg_{G}(v)$ times for each flow $\ff_i$. To achieve this behaviour, we maintain a link-cut tree per flow $\ff_i$. If edge $(u, v)$ carries flow from $u$ to $v$, we look for the edge $e'$ from $u$ to the root carrying the minimum amount of flow (break ties by taking the edge closest to $u$). If $e'$ carries $0$ flow we remove it from the tree. Then, we remove $-1$ flow from the path of $u$ to its (new) root $r$. Next, we check if $r$ has an incoming edge carrying flow in $\ff_i$. If so, we add this edge to the link cut tree (this case corresponds to erroneously backtracking) and increment a counter $\qq_i(r)$ if the vertex $r$ is in $\widehat{S}_i$.

Otherwise, we have found a true source $r$ with respect to the flow $\ff_i$. If $r \in \widehat{S}_i$, we check if $r$ has incoming flow from a flow $\ff_j$ with $j < i$ by going through all the incident edges and flows in a fixed but arbitrary order (evaluating one edge and flow combination every time $r$ is found as a true source). If so, we recursively backtrack $\ff_j$. Otherwise, we increase $\qq_i(r)$ and rule out one combination of edges and flow. If $\qq_i(r) > (1 + i)\dd(r)$, the vertex $r$ gets pruned, that means it's added to $S_0$ and therefore removed from $\widehat{S}_i$. Notice that $\qq_i(r) > (1 + i)\dd(r)$ ensures that we attempted to backtrack every possible edge and flow combination for vertex $r$. 

To summarize, the flow maintenance data structure consists of an initialization routine, a re-initialization routine that re-initializes the sets $S_l, \ldots, S_k$ and a routine for removing edges, which first internally removes all the flow using the procedure described above, and then deletes the edge. The pseudocode in \Cref{alg:flow_cert_maint} contains a more detailed description of the algorithm.

\subsection{Analysis of Batch Flow Maintenance}

In this section, we clarify some details of the flow maintenance data structure and we prove \Cref{thm:flow_cert_maint}. We first show that the call to the local blocking flow algorithm $\textsc{Dinitz}()$ computes an exact maximum flow in \Cref{alg:flow_cert_maint:dinitz} of \Cref{alg:flow_cert_maint}  because routing flow on an expander requires significantly less calls to a blocking flow routine than routing on a general graph. We remark that this routine is the only step in the algorithm presented in this section that is amortized, and we de-amortize this final ingredient in \Cref{sec:worst_case_low_rec}.  

\begin{claim}
    \label{clm:exact_max_flow}
    Given the settings of \Cref{thm:flow_cert_maint}, the blocking flow algorithm of Dinitz in \Cref{alg:flow_cert_maint:dinitz} of \Cref{alg:flow_cert_maint} computes an exact maximum flow routing the demands in time $\tilde{O}(\vol_G(S_i)/\phi)$.
\end{claim}
\begin{proof}
    Recall that the distance from every vertex in $S_i$ to a non-saturated source is at least $h$ after running blocking flow for $h$ iterations. We let $S_{\leq i}$ denote the set of vertices that are reachable with a $i$ hop path from a vertex with excess source on it. Consider the number of edges in the cut $(S_{\leq i}, V \setminus S_{\leq i})$ in the residual graph. Because of the expansion of $\widehat{G}_i$ being $\phi/10$, there are $\frac{\phi}{10} \vol_G(S_{\leq i})$ edges in this cut except for the ones that are saturated. They can be saturated for two reasons: 1) flow coming into $S_{\leq i}$ and leaving again or 2) flow emitted by sources in $S_{\leq i}$. 1) can reduce the number of edges by at most half, and since out edges are saturated 2) can reduce the amount of edges by at most their capacity divided by the total source. Since the total source is at most $(10 k^2 + k + 1)\vol(S_{\leq i})$ but the total capacity is at least $800 k^3 \vol(S_{\leq i})$ at most say $1/4$ of the out edges can be saturated for this reason. Therefore there are at least $\phi \vol_G(S_{\leq i})/24$ edges in the cut, and therefore $\vol_G(S_{\leq i + 1}) \geq (1 + \frac{\phi}{24})\vol_G(S_{\leq i})$. But this means that the volume of $S_{\leq h}$ exceeds the volume of the graph. Notice that this means that the demands are routed, by the bound on the volume of $S_i$ and thus the total amount of source. 

    The runtime follows from \Cref{fct:dinitz}.
\end{proof}

We then prove the simple observation that the sinks remain bounded per flow $\ff_i$.

\begin{claim}
    \label{clm:sink_inc}
    For $i = 1, \ldots, k$ every vertex $v \in V \setminus \bigcup_{j = 1}^i S_j$ initially absorbs $(10 \cdot k - 1) \deg_{G}(v) = (10 \cdot k - 1) \dd(v)$ flow from $\ff_i$. Then, this absorption increases by at most $\deg_{G}(v)$ until $\ff_i$ gets re-initialized. 
\end{claim}
\begin{proof}
    The absorption per flow is bounded by $(10 \cdot k - 1) \cdot \deg_{G}(v)$ initally by \Cref{clm:exact_max_flow} and the parameters of the flow problem in \Cref{alg:flow_cert_maint:dinitz} of \Cref{alg:flow_cert_maint}. 

    Then, by the definition of our algorithm the only way for a vertex $v$ to absorb one more flow from $\ff_i$ is if the condition in  \Cref{alg:flow_cert_maint:link} evaluates to true during a call to $\textsc{RemoveFLow}(\cdot, i)$. But this condition can evaluate to true at most once per edge incident to vertex $v$, and therefore the claim follows.
\end{proof}

We then prove that the data structure maintains flow certificates as required. 

\begin{invariant}
    \label{lem:cert_maint}
     Given the settings of \Cref{thm:flow_cert_maint} and $\widehat{G}_i := \widehat{G}[V \setminus \bigcup_{j = 0}^{i - 1} S_j]$ the data structure \Cref{alg:flow_cert_maint} maintains that $\left(S' = \bigcup_{i = 1}^k \widehat{S}_i, \ff\right)$ remains a $(10 \cdot k, 10 \cdot k^2, 8000 k^4/\phi)$ flow certificate on $\widehat{G}_1$.
\end{invariant}
\begin{proof} 
    We show congestion, sink capacity upper bounds and source capacity lower bounds separately. 
    \begin{itemize}
        \item \underline{Congestion:} The congestion bound of $8000 k^4/\phi$ for $\ff$ follows directly because every individual flow $\ff_i$ never congests an edge by more than $8000 k^3/\phi$. This can be seen from the fact that the initial threshold upon re-initialization is $8000 k^3/\phi$ and flows only decrease the amount they route thereafter until they get re-initialized again by the definition of our algorithm. 
        \item \underline{Sink capacity upper bound:} Every vertex $v$ absorbs at most $10k \cdot \deg_{G}(v)$ flow from each flow $\ff_i$ throughout by \Cref{clm:sink_inc}. Therefore, the bound of $10 \cdot k^2$ on the sinks follows directly.
        \item \underline{Source capacity lower bound: } We show that every vertex $v \in S'$ routes out at least $10k \cdot \deg_{G}$ flow in $\ff$. We proceed by induction and assume that vertices in $v$ in $\widehat{S}_1, \ldots, \widehat{S}_{j - 1}$ route out at least $10 \cdot \deg_{G}(v)$. We then consider the set $\widehat{S}_j$ directly after re-initialization. Again, by \Cref{clm:sink_inc} the vertices $u \in S_j$ absorb at most $(j - 1) 10 k \cdot \deg_{G}(u)$ flow from the flow out of the sets $\widehat{S}_1, \ldots, \widehat{S}_{j - 1}$ throughout. The flow $\ff_j$ initially routes out $(10jk + j  + 1) \cdot \deg_{G}(u)$ flow by the definition of the flow problem in \Cref{alg:flow_backtracking:dinitz} of \Cref{alg:flow_backtracking}. Therefore, every vertex in $S_j$ routes out $(10k + j + 1) \cdot \deg_{G}(u)$ flow that is not cancelled by flow certificates of lower index. 
        We note that when the flows/sets at levels $i < j$ get re-initialized, so does level $j$ by the description of our algorithm. Therefore, all the flows $\ff_i$ for such level only decrease in support until layer $j$ gets re-built. 

        We now show that whenever a vertex $v \in \widehat{S}_j$ routes out one less flow, the counter $\qq_j(v)$ is incremented. Every vertex $v \in \widehat{S}_j$ can enter the if condition at \Cref{alg:flow_cert_maint:link} at most $\deg_{G}(v)$ times, and therefore such re-links can contribute at most $\deg_{G}(v)$ to $\qq_j(v)$. Clearly, each such re-link causes the vertex to route out one less flow. Thereafter, the algorithm adds $1$ to $\qq_j(v)$ whenever $v$ is backtracked to as a true source of $\ff_j$, and no flow into $v$ coming from $\ff_i$ for $i < j$ can be removed (even though there might be some flow coming in). Notice that no flows $\ff_l$ for $l > j$ can be adjacent to $v$ since they are computed on a graph that does not include the vertices in $\widehat{S}_j$. In that case $\qq_j(v)$ also gets incremented. 

        Whenever $\qq_j(v)$ has been increased by another $(j - 1) \cdot  \deg_{G}(v)$ there cannot be any edge $(x, v)$ in some $\dir{E}_i$ for $i < j$ since one such edge was removed every time $\qq_j(v)$ was increased by one, and the total number of such edges is $j \cdot \deg_{G}(v)$. Therefore, whenever a vertex is pruned due to $\qq_j(r)$ reaching $(j + 1) \cdot \deg_{G}(r)$ it is a true source of the flow $\ff$, and in fact has already lost $\deg_{G{(0)}}$ out-flow as a true source. 

        Finally, since the algorithm adds $1$ to $\qq_j(v)$ whenever it routes out one less source, and it routes out at least $(10k + j + 1) \cdot \deg_{G}(u)$ flow initially, every non-pruned vertex still routes out at least $10k \cdot \deg_{G}(v)$ flow  since every vertex for which $\qq_j$ reaches $(j + 1)\cdot \deg_{G}(r)$ is pruned. 
    \end{itemize}
    This concludes the proof.     
\end{proof}

We finally prove the main theorem which follows rather directly. 

\begin{proof}[Proof of \Cref{thm:flow_cert_maint}]
    The expansion guarantee of $\widehat{G}[V \setminus S_0]$ follows from \Cref{lem:cert_maint} and \Cref{clm:exact_max_flow}. 

    The runtime guarantee for the re-initialization follows from \Cref{clm:exact_max_flow} and the fact that a flow decomposition can be done in time proportional to the total flow weight. 
    
    The runtime guarantee for the edge deletions follows from \Cref{algo:LCT} and the fact that every call to $\textsc{DeleteEdge}()$ causes at most $O(k^4/\phi)$ calls to dynamic tree data structures. 
\end{proof}

\subsection{A Low Recourse Algorithm for Expander Pruning}

In this section, we use the flow certificates developed in the previous section to show that the algorithm presented in \Cref{sec:amortized} can be turned into a low recourse pruning algorithm.  

\paragraph{Overview.} We recall the batching scheme presented in \Cref{sec:batching} that turned \Cref{alg:batch_pruning} into a dynamic algorithm for expander pruning with low amortized recourse. It maintains all the edge deletions to $G$ that occurred so far in batches $B_1, B_2, \ldots, B_k$ such that batch $B_i$ either contains $2^{k - i}$, $2^{k - i - 1}$, or $0$ edges (we call such a batch full, half-empty and empty respectively).\footnote{Here we technically only describe the set of deletions $B_i$. It additionally considers some pruned sets $A_i$ to make sure that the pruned set only grows over time. We can safely ignore this technicality for the discussion in this section, but the full algorithm requires these additional deletions to ensure the monotonicity of the pruned set. Our worst-case recourse and update time algorithm again discusses this technicality in detail.} Then, whenever a new deletion $e$ occurs, the algorithm finds the empty or half-full batch with highest index, e.g. $B_i$ and lets $B_i \gets \bigcup_{j > i} B_i \cup \{e\}$. Finally, it sets $B_j \gets 0$ for $j > i$. Re-starting \Cref{alg:batch_pruning} at index $i$ then yields new sets $S_i, \ldots, S_{\lambda = k + \frac{1}{\phi} + 1}$ such that $G[V \setminus \bigcup_{i = 1}^{\lambda} S_i]$ remains a $\phi/10$-expander. 

We first notice that whenever layer $i$ gets rebuilt in the amortized pruning algorithm presented in \Cref{sec:batching}, it takes another $2^{k - i}$ deletions before any layer $i' \leq i$ gets re-built. Therefore, we aim to spread out the removal of $S_i$ over these deletions. Since $\vol_G(S_i) \leq O(\lambda \cdot \log(n) \cdot 2^{k - i}/\phi)$ by \Cref{clm:source_reduction}, removing $O(\lambda \log(n)/\phi)$ vertices per deletion should suffice. The tricky part is to remove vertices such that the intermediate graphs are $\tilde{\Omega}(\phi)$ expanders too. As the reader might anticipate, this is where the batch flow certificates introduced in this section come into play.

\paragraph{Delayed pruning algorithm. } Throughout the execution, our dynamic algorithm additionally maintains a batch flow certificate $\mathcal{D} \gets \textsc{BatchFlowCert}()$ (\Cref{alg:flow_cert_maint}) that gets updated whenever the algorithm presented in \Cref{sec:amortized} updates the sets $S_i, \ldots, S_k$, where for simplicity we let the set $S_k$ additionally contain all the vertices pruned by layers of larger index than $k$, i.e. the contents of the sets $S_{k + 1}, \ldots, S_{\lambda}$.

Then, we let the maintained decremental pruned expander be $G[V \setminus S_0]$, i.e. the incremental set $S_0$ contains all the vertices that the algorithm decided to prune so far. 

Whenever an edge deletion occurs, for every $i$ we first delete $10^6 \lambda \log_2(n)/ \phi$ arbitrary edges $e$ incident to $S_i$ in the current graph $\widehat{G}[V \setminus \cup_{j <i} S_j]$ using $\mathcal{D}$, as long as such edges still exist, i.e. we delete some volume of edges of each set $S_i$ that \Cref{alg:batch_pruning} would have already considered as pruned. \footnote{In the overview, we described removing the boundary edges instead of the full volume. It is necessary to remove the whole volume instead when using blocking flow.} Then, whenever a rebuild of \Cref{alg:batch_pruning} happens, we re-initialize all the re-built sets in $\mathcal{D}$, where we set $S'_i \gets S_i$ and $S'_k \gets \bigcup_{j \geq k} S_j$, i.e. the last set includes the extra pruning due to empty batches. Finally, we forward the actually deleted edge to $\mathcal{D}$ and remove it from the graph. 

\paragraph{Proof of Correctness. }

We first formalize the notion of a rebuild. 

\begin{definition}[Rebuild] 
    \label{def:rebuild}
    We say a level $i$ is rebuilt, if the for-loop at \Cref{alg:batch_pruning:main_loop} is run with index $i$. In particular, whenever level $i$ is rebuilt, so are levels $i + 1, \ldots, k + 10$. 
\end{definition}

We prove the central theorem of this section. 
\lowRecourseAmortTime*

\begin{proof}
    We first show that the algorithm calls the data structure $\mathcal{D}$ in a way that conforms with \Cref{thm:flow_cert_maint}. The constraints on the volume of $S$ is satisfied by the bound on the total number of deletions and  \Cref{clm:volume}. Furthermore, we have that $\vol_{\widehat{G}_i}(S_i) \leq \vol_G(S_i) \leq 2 \cdot (64 \log_2 n) 100 \cdot \lambda \cdot 2^{k - i}/\phi$ for all $i$ by \Cref{clm:source_reduction} and a geometric series for $S_k$ since we add all the vertices pruned by levels $k + 1, \ldots, \lambda$ in \Cref{sec:batching} to $S_k$. Therefore, all these edges are deleted by the time layer $i$ gets rebuilt by \Cref{clm:rebuild_timing} and the description of our algorithm. But, since $\widehat{G}$ remains an expander it certainly cannot contain isolated vertices, and therefore $S_i$ has to be fully contained in $S_0$ at that stage. This, in particular means that any vetex that is in some set $A_i$ is always contained in $S_0$. The requirement on the rebuilt sets being empty is thus satisfied as well. Finally the graph $\widehat{G}[V \setminus S]$ remains a $\phi/10$-expander by \Cref{lem:maintains_exp} since we apply exactly the same batching scheme. 

    Given that the update sequence fulfills these constraints, the theorem directly follows from \Cref{thm:flow_cert_maint} and the description of our algorithm where $k = \ceil{\log_2 m}$ and $\lambda = \tilde{O}(1)$ because $G$ is connected and therefore $\phi \geq 1/m$. 
    The amortized runtime guarantee finally follows from the fact that the meta algorithm from \Cref{sec:amortized} runs in time $\tilde{O}(D/\phi^2)$ by \Cref{clm:runtime} and the extra work introduced by the flow certificate is bounded by $\tilde{O}(D/\phi^2)$ via the bound on the volume of the pruned sets by the meta algorithm by \Cref{clm:volume}, and \Cref{thm:flow_cert_maint}. 
\end{proof} 

\section{Worst-Case Time Low Recourse Expander Pruning}
\label{sec:worst_case_low_rec}

In this section, we show that by careful scheduling, the batching based pruning algorithm developed in the previous sections allows us to obtain the best of both worlds: slow pruning for worst-case recourse with simultaneous worst-case update time guarantees. To do so, we first show that a flow certificate can be initialized on a decremental graph. Since computing a flow certificate was the only step requiring amortization in \Cref{sec:amortized_low_rec}, this then allows us to build an adapted batching scheme with low worst-case time. 

\subsection{Flow Certificate Repairing via Nesting}

Since we intend to use flow certificates in a worst-case decremental environment, it is only natural to spread the initialization of a flow certificate over $\tilde{\Omega}(\vol_G(S))$ deletions to adequately bound the compute spent after each individual deletion. However, then the underlying graph is undergoing edge (and vertex) deletions while a flow certificate is being initialized. Namely, some large fraction of the edges routing flow in a certificate could be removed from the graph by the time said flow has finished being computed. Fortunately, we can repair a flow certificate with little overhead using the nesting property introduced in \Cref{fct:flow_cert_comp}. Furthermore, the re-initialization is the only operation for which the algorithm presented in \Cref{sec:amortized_low_rec} uses more than $\tilde{O}(1/\phi^2)$ compute steps, and therefore resorts to amortization. 

We first extend the definition of pruned expander graphs to allow measuring with respect to some externally provided degree vector $\dd$ that is an overestimation of the actual degrees. Notice that $\dd \geq \deg_{G}$, and therefore this is a strengthening of the definition of pruned expander graphs (\Cref{def:pruned_exp}).

\begin{definition}
    We call a decremental graph $G = (V, E)$ a $\phi$-(decremental pruned) expander with respect to $\dd \in \R^V$ if $\dd \geq \deg_{G}$ and at any stage for all $S \subset V$ so that $\vol_{\dd}(S) \leq \vol_{\dd}(S \setminus V)$ we have $|E_{G^{(i)}}(S, V \setminus S)| \geq \phi \cdot \vol_{\dd}(S)$ where $\vol_{\dd}(S) := \sum_{v \in S} \dd(v)$.
\end{definition}

In our algorithm, we will always set $\dd$ to the degree vector of the input graph. However, we use the data structure developed in this section to initialize flows on a already partially pruned subgraph of the input graph, and we therefore introduce this extra parameter for convenience. 

\paragraph{Limiting the number of deletions. } In the remainder of this article, we make use of the assumption that the total number of deletions is bounded by $\phi m/\polylog(m)$ for some suitably large poly-logarithmic factor. This ensures that all the sets $S_i$ ever pruned are the smaller side of the corresponding cut in the remaining graph with a large poly-logarithmic gap. Thus, the total source capacity always exceeds the total sink capacity. Although we keep track of this assumption in our statements, the reader might want ignore this technicality when initially reading this section.

\paragraph{Worst-case flow initialization.} We state the main lemma of this subsection. It yields a flow with the  same guarantees as the call to $\textsc{Dinitz}()$ in \Cref{alg:flow_cert_maint:dinitz} of \Cref{alg:flow_cert_maint}, and therefore allows smearing out the computation that $\textsc{Dinitz}()$ performed in \Cref{alg:flow_cert_maint}. This is the only amortized part of \Cref{alg:flow_cert_maint}. 

\begin{lemma}
    \label{lem:init_spread}
    The data structure $\textsc{WorstCaseFlow}(G = (V, E), S, \dd, \phi)$ is initialized with 
    \begin{enumerate}
        \item a $\phi/10$-expander $G$ with respect to $\dd$,
        \item a set $S$ such that $\vol_{\dd}(S) \leq (128 \log_2 n) \cdot 100\lambda  \cdot 2^{k - l} / \phi$,
        \item and then processes batches $\{P_l, \ldots, P_{k} \}$ and $\{U_l, \ldots, U_{k}\}$ of edge and vertex deletions $P_i \subseteq E$ and $U_i \subseteq V$ to $G$ such that
        \begin{itemize}
            \item $G$ remains a $\phi/10$-decremental pruned expander with respect to $\dd$ throughout
            \item and $ \vol_{\dd}(V \setminus S)\geq \vol_{\dd}(S)/(50k^2)$ ($S$ firmly on smaller side of the cut).
        \end{itemize}
        Furthermore, $|P_i| \leq 2^{k - i}$ and $\vol_{\dd}(U_i) \leq 2 \cdot 10^7 \log_2^3 n \cdot 2^{k - i} / \phi$ for $\lambda = k + 1 + \log_2 \phi^{-1} = O(\log m)$ and $k = \ceil{\log_2 m}$.
    \end{enumerate}
    Then, the algorithm is initialized in time $\tilde{O}(k^3 \cdot \vol_{\dd}(S)/\phi)$ and processes each batch $(P_i, U_i)$ in time $\tilde{O}(k^4 \cdot 2^{k - i}/\phi^2)$ and finally returns a flow $\ff$ that routes at least $(10l \cdot k + l + 1) \cdot \dd[S \setminus U]$ source to sinks of capacity at most $(10k - 1) \dd[V \setminus (S \cup U)]$ with edge capacity $8000 k^3/\phi$ where $U : = \bigcup_{i = 1}^k U_i$.
\end{lemma}

\paragraph{Algorithm overview. } In the following algorithm description and analysis, we denote with $G^{(i)} = (V^{(i)}, E^{(i)})$ the graph $G[V \setminus \bigcup_{j = l}^i U_j] \setminus \bigcup_{j = l}^i P_j$, i.e. the graph $G$ after removing batches $P_l, \ldots, P_i$ and $U_l, \ldots U_i$, and $G^{(0)} = G$.

The $\textsc{WorstCaseFlow}()$ algorithm (See \Cref{alg:worst_case_flow_cert}) leverages the $\textsc{FlowBacktracking}()$ algorithm (See \Cref{alg:flow_backtracking}) which is very similar to the $\textsc{BatchFlowCert}()$ algorithm developed in \Cref{sec:amortized_low_rec}. 

\paragraph{Flow backtracking. }The flow backtracking algorithm $\textsc{FlowBacktracking}$ (\Cref{alg:flow_backtracking}) is given a graph $G = (V, E)$ alongside a set $S \subseteq V$, a vector $\dd$ that over-estimates the degrees, and parameters $\phi$ and $\theta$. The parameter $\phi$ is used to quantify the expansion, and $\theta$ quantifies the amount of flow we want to route out of the set $S$. 

It initializes a cycle-free flow $\ff$ routing $(11\theta + 1) \dd[S]$ source flow to sinks of capacity $8\dd[V \setminus S]$. Then, it receives batches $(P_i, U_i)$ of vertex and edge deletions. Thereafter, it removes all the flow on deleted edges in $P_i$ and edges incident to deleted vertices in $U_i$ via backtracking as in \Cref{alg:flow_cert_maint}. However, instead of pruning vertices that send out too little flow, it merely returns these vertices by adding them to $S_{\text{rerun}}$. Furthermore, it does not need to worry about interactions between multiple nesting flows in comparison to the more complicated flow-backtracking algorithm presented in \Cref{sec:amortized_low_rec}.

\paragraph{Worst-case flow algorithm.} The $\textsc{WorstCaseFlow}()$ algorithm (See \Cref{alg:worst_case_flow_cert}) first initializes a $\textsc{FlowBacktracking}()$ data structure routing out of the set $S$. It then processes one batch of deletions after another, and initializes an extra $\textsc{FlowBacktracking}()$ data structure for the union of the returned vertex sets after each deletion batch. It makes sure that all $\textsc{FlowBacktracking}()$ data structures route out $10k \dd(v)$ extra flow per soruce $v$ by setting their $\theta$ parameter to $10kl + 10k + l + 1$, and therefore the flows stored in these data structures add up to the desired flow $\ff$ routing out $(10kl + l + 1)\deg(v)$ in the end since the sink capacity of every vertex $v$ per flow is bounded by $9\deg(v)$ and the source is at least $\theta \deg(v)$ for vertices $v \in S cap U$. 

The pseudocodes \Cref{alg:worst_case_flow_cert} and \Cref{alg:flow_backtracking} contain a detailed description of the algorithms.

\begin{algorithm}
\caption{\textsc{WorstCaseFlow}($G^{(l - 1)} = (V^{(l - 1)}, E^{(l - 1)}), S, \dd, \phi, \{(P_l, U_l), \ldots, (P_k, U_k)\}$)}
\label{alg:worst_case_flow_cert}
    $S'_l \gets S$ \\ 
    \For{$i = l, \ldots,  k - 1$}{
        \tcp{Process the next batch}
        $\mathcal{D}_i \gets \textsc{FlowBacktracking}()$ \\
        $\mathcal{D}_i.\textsc{Initialize}(G^{(i - 1)}, S'_i, \dd, \phi, 10(l \cdot k + k) + l + 1)$ \\
        \For{$j = l, \ldots, i$}{
            $S^{j, i}_{\text{rerun}}, \ff_j \gets \mathcal{D}_j.\textsc{RemoveBatch}(P_i, U_i)$ \\
        }
        $S'_{i + 1} \gets V^{(i)} \cap \bigcup_{j = l}^i S^j_{\text{rerun}}$ \tcp*{Non pruned vertices that route out to little}
    }
    $\mathcal{D}_{k} \gets \textsc{FlowBacktracking}()$ \\
    $\ff_{k} \gets \mathcal{D}_{k + 1}.\textsc{Initialize}(G^{(k)}, S'_{k + 1}, \dd, \phi, k )$ \\
    \Return{$\sum_{i = l}^k \ff_i$}
\end{algorithm}

\begin{algorithm}
\caption{\textsc{FlowBacktracking}()}
\label{alg:flow_backtracking}
\SetKwFunction{proc}{\textsc{RemoveFlow}}\SetKwFunction{procbatch}{\textsc{RemoveBatch}}
\SetKwFunction{procinit}{\textsc{Initialize}}

\SetKwProg{myproc}{Procedure}{}{}

\myproc{\procinit{$G = (V, E), S, \dd, \phi, \theta$}}{
    $\ff \leftarrow \textsc{Dinitz}((V, E, \frac{400 \theta}{\phi}), \ss = (\theta + 1) \cdot \dd[S], \tt = 8 \cdot \dd [V \setminus S], h = \frac{200}{\phi} \log_2 n)$ \label{alg:flow_backtracking:dinitz}\\   Compute a path-cycle decomposition of $\ff$, and update $\ff$ by removing all cycles. \\
    $\dir{G} \gets (\dir{V} = V, \dir{E} = \supp(\ff))$ \\
    $\mathcal{D} \gets \textsc{DynTree}.\textsc{Initialize}(\dir{G}, \ff)$ \tcp*{Initialize empty tree.} 
    $\widehat{G} \gets G$; $\qq \gets \veczero_V$ \\
    \Return{$\ff$}
}

\myproc{\procbatch{$P, U$}}{
    $P' \gets P \cup \widehat{E}(U, \cdot)$ \tcp*{Add all edges adjacent to $U$ to be removed. }
    $\widehat{G} \gets (\widehat{G} \setminus P)[\widehat{V} \setminus U]$ \tcp*{Remove batch of vertices and edges from graph. }
    $S_{\text{rerun}} \gets \emptyset$ \\
    \ForEach{$e = (u,v) \in P'$}{
        \For{$j = 1, \ldots, 400\theta/\phi$}{
            \If{$(u,v) \in \dir{E}$ or $(v, u) \in \dir{E}$}{
                Assume wlog that $e = (u,v) \in \dir{E}$. \\
                    \If{$\mathcal{D}.\textsc{ReadCurrentFlow}(e) > 0$}{
                        \If{$\mathcal{D}.\textsc{ReadCurrentFlow}(\mathcal{D}.\textsc{FindMin}(v)) = 0$}{
                        $e' = (x,y) \gets \mathcal{D}.\textsc{FindMin}(v)$ \\
                        $\mathcal{D}.\textsc{Delete}(e')$; $\dir{G} \gets \dir{G} \setminus \{e'\}$ \\
                        }
                        $\mathcal{D}.\textsc{UpdateFlow}(v, - 1)$\\
                        $r \gets \mathcal{D}.\textsc{FindRoot}(v)$\label{alg:flow_backtracking:root} \\ 
                        \If{$(x, r)$ in $\dir{E}$}{
                            $\mathcal{D}.\textsc{Insert}(x,r)$
                        }
                        \If{$r \in S$}{
                            $\qq(r) \gets \qq(r) + 1$ \\
                            \If{$\qq(r) > \dd(v)$}{
                                \tcp{Give up on vertex $r$}
                                $S_{\text{rerun}} \gets S_{\text{rerun}}
                                \cup \{r\}$ 
                            }
                        }
                    
                }
            }
        }
    }
    \Return{$S_{\text{rerun}}, \ff$}
}
\end{algorithm}

\paragraph{Analysis. } We first analyse the $\textsc{FlowBacktracking}()$ data structure (\Cref{alg:flow_backtracking}) in an analogous way to the analysis of \Cref{alg:flow_cert_maint}. 

\begin{lemma}
    \label{lem:flow_backtracking}
    For any $l \in [k]$, when initialized with 
    \begin{itemize}
        \item a graph $G^{(l - 1)} = (V^{(l - 1)}, E^{(l - 1)})$, a parameter $\phi$ and vector $\dd$ such that $G^{(l - 1)}$ is a $\phi/10$ decremental pruned expander with respect to $\dd$ and
        \item a set $S \subset V$ and some integer $\theta$
    \end{itemize}
    the data structure $\textsc{FlowBacktracking}()$ (\Cref{alg:flow_backtracking}) processes a sequence of batches $(P_i, U_i)$ for $i = l, \ldots, k$ such that 
    \begin{itemize}
        \item $G^{(i)}$ remains a $\phi/10$ decremental pruned expander,
        \item $8\vol_{\dd}(S \cap V(G^{(i)})) \leq \vol_{\dd}(V(G^{(i)}) \setminus S)/(\theta + 1)$
        \item $|P_i| \leq 2^{k - i}$ and $\vol_{\dd}(U_i) \leq (128 \log_2 n) \cdot 100\lambda \cdot 2^{k - i} / \phi$
    \end{itemize}    
    After processing all batches it explicitly maintains a set $S_{\text{rerun}}$ and a $(\theta, 10, 400\theta \cdot k/\phi)$ flow certificate $(V \cap (S \setminus S_{\text{rerun}}), \ff)$. Furthermore, $\vol_{\dd}(S_{\text{rerun}}^i) \leq 400 \theta \cdot k \cdot 132 \cdot 2^{k - i}/\phi$ where $S_{\text{rerun}}^i$ is the set of vertices returned after processing batch $i$ and $S_{\text{rerun}} = \bigcup_{i = l}^k S_{\text{rerun}}^i$. 
    
    The processing time per batch $(P_i, U_i)$ is $\tilde{O}(\frac{\theta k}{\phi} (|P_i| + \vol_{\dd}(U_i)))$ and the initialization time is $\tilde{O}(\theta \vol_{\dd}(S)/\phi)$.
\end{lemma}
\begin{proof}
    We first show that Dinitz algorithm in \Cref{alg:flow_backtracking:dinitz} of \Cref{alg:flow_backtracking} computes a flow routing the demands. The proof is completely analogous to the proof of  \Cref{clm:exact_max_flow}. 

    After running Dinitz algorithm for $h = \frac{200}{\phi} \log_2 n$ iterations, the distance from every vertex in $S$ to a vertex in $V^{(l-i)} \setminus S$ in the residual graph is at least $h$. Let $S_{\leq j}$ denote the number of vertices that can be reached from $S$ in at most $j$ hops from $S$ in the residual graph. The total number of edges in $G^{(l - 1)}$ is at least $\frac{\phi}{10} \vol_{\dd}(S_{\leq j})$ by expansion of $G^{(l - 1)}$ with respect to $\dd$. At most half the saturation can be due to flow entering the set $S_{\leq j}$. But the total amount of source is at most $(\theta + 1)\vol_{\dd}(S_{\leq j})$ and every edge has capacity $400\theta/\phi$. Therefore, at least a quarter of the edges remain in the cut, and therefore $\vol_{\dd}(S_{\leq j + 1}) \geq (1 + \phi/40)\vol_{\dd}(S_{\leq j}$. For $j = \frac{200}{\phi} \log_2 n$ the set $S_{\leq j}$ contains all vertices in $V^{(l - 1)}$, and therefore the flow is routed since the total sink always exceeds the total source demand.

    We then argue that the algorithm maintains a $(\theta, 10, 400\theta \cdot k/\phi)$ flow certificate $(S \setminus S_{\text{rerun}}, \ff)$ analogously to the proof of \Cref{lem:cert_maint}. The algorithm backtracks and removes flow and adds a vertex $v \in S$ to $S_{\text{rerun}}$ whenever it emits $\dd(v)$ less flow than it started out with (and it did not get pruned). Therefore it still emits $\theta$ flow as long as it has not entered $S_{\text{rerun}}$. Furthermore, every source $v$ initially absorbs $8 \dd(v)$ source. Every time it absorbs one extra flow it gets re-connected, and since there are at most $\dd(v)$ edges carrying flow into $v$ it absorbs at most $9 \dd(v)$ flow throughout. 

    The runtime of the initialization follows from \Cref{fct:dinitz}. The runtime of the $\textsc{RemoveBatch}()$ routine is proportional to $400 \theta k/\phi (\vol_{\dd}(U) + |P|)$ and therefore the runtime follows. Furthermore, the size of the set $S_{\text{rerun}}^i$ is proportional to $\frac{400 \theta k}{\phi} (\vol_{\dd}(U) + |P|)$. This concludes the proof. 
\end{proof}

Given the previous lemma, it is easy to prove the main lemma of this section. 

\begin{proof}[Proof of \Cref{lem:init_spread}]
By \Cref{lem:flow_backtracking} we have that $S_{\text{rerun}}^{i,j}$, i.e. the $i$-th set $S_{\text{rerun}}$ output by data structure $\mathcal{D}_j$, has volume $\vol_{\dd}(S_{\text{rerun}}^{i, j}) \leq 400 \theta \cdot k \cdot 132 \cdot 2^{k - i}/\phi$, and therefore the volume $\vol_{\dd}(S'_{i + 1}) \leq 400 \theta \cdot k^2 \cdot 2 \cdot 66 \cdot 2^{k - i}/\phi $ is bounded. Notice that this volume is also upper bounded by the volume of $S$ and therefore firmly on the small side of the cut. 

Then, the runtime guarantees follow directly from the description of the algorithm and \Cref{lem:flow_backtracking}. 

Since the sink capacity of a vertex $v$ never exceeds $9k \deg_{G^{(0)}}(v)$, and every source $v$ is guaranteed to continue sending out $(10 k \cdot l + 10k + l + 1)\deg_{G^{(0)}}(v)$ flow, the net out-flow of every source is still $(10 k \cdot l + l + 1)\deg_{G^{(0)}}(v)$ as desired. The edge capacity is at most $k \cdot 400 \theta/\phi \leq 8000 k^3/\phi$ for large enough $k$, and the sink capacity is at most $9k \deg_{G^{(0)}}(v) \leq (10k - 1)\deg_{G^{(0)}}(v)$ as argued above. 

This concludes the proof.
\end{proof}

\subsection{A Low Recourse Worst-Case Update Time Pruning Algorithm}

In this section, we use the developed worst-case flow certificates to finally get low recourse and a worst-case update time guarantee simultaneously. 

\paragraph{Data structure description.}

As previously in \Cref{sec:amortized} and \Cref{sec:amortized_low_rec}, the worst-case recourse algorithm $\textsc{WorstCaseLowRec}()$ also processes the edge deletions in batches with the batching scheme presented in \Cref{sec:batching}. 

We recall that whenever a batch becomes half-full for the first time, it will never be empty again, and that edges always move from a batch $B_j$ to a batch $B_{j - 1}$. Furthermore, when batch $B_j$ becomes full, all the batches $B_{j + 1}, \ldots, B_{k}$ are half-full after. 

\paragraph{Pre-computing Framework. } We first discuss how this rebuilding scheme is tailored to the worst-case recourse and update time setting by allowing a lot of distribution of computational cost. The next update of a half-full or empty batch $B_j$ can be computed as soon as the batch $B_{j + 1}$ becomes full. At that stage, however, the batches $B_{j + 2}, \ldots, B_{k}$ are all half full, and therefore a lot more insertions have to arrive before the update of batch $B_j$ has to be ready. We use this time to distribute the cost of computing this update.

\begin{itemize}
    \item \underline{Notation and Rebuild Requirements:} To accurately describe our rebuilding scheme, we first introduce some notation. As described above, we always run the meta algorithm (\Cref{alg:batch_pruning}) on batches $B_i, A_i$, where the set $B_i$ contains deletions as described above, and the set $A_i$ contains additional pruned vertices to ensure monotonicity of the pruned set. The exact sets $A_i$ are described in detail below and in \Cref{sec:batching}. We denote with $S_i$ the set that the meta algorithm proposes to prune at level $i$ in \Cref{alg:batch_pruning:proposal} of \Cref{alg:batch_pruning} and with $S_0$ the vertices that are already pruned as in \Cref{sec:amortized_low_rec}. We then denote $\widehat{S}_i := S_i \setminus S_0$, i.e. the sets that are proposed to be pruned but have not yet been pruned. We maintain that there is a flow $\ff_i$ routing at least $10ik \dd[\widehat{S}_i]$ source to sinks of capacity at most $10k\dd[V \setminus \bigcup_{l = 0}^i S_l]$ in $G_i := G[V \bigcup_{l = 0}^{i - 1} S_l]$ where $G$ is a decremental graph and thus experiences deletions. As before, we maintain that every vertex that is in a previously pruned set $A_i$ is also in $S_0$. 

    Every layer $i \leq k - 2$ can be pre-computing a rebuild of the layers $i, \ldots, k$ at any stage. We refer to the relevant variables $X$ of layer $j$ being precomputed by a rebuild of layer $i$ as $X^{(i)}_j$, i.e. $B_j^{(i)}$ denotes the set of  deletions at layer $j$ that a rebuild at layer $i$ proposes. When pre-building the flows $\ff_j^{(i)}$ we ensure that they route at least $(10jk + j + 1) \dd\left[\widehat{S}_j^{(i)}\right]$ source to sinks of capacity at most $(10k - 1)\dd\left[V \setminus \left(\bigcup_{l = 0}^{i - 1} S_l \cup \bigcup_{l = i}^j S_l^{(i)}\right)\right]$ in $G_j^{(i)} := G\left[V \setminus \left(\bigcup_{l = 0}^{i - 1} S_l \cup \bigcup_{l = i}^{j - 1} S_l^{(i)}\right)\right]$. See  item 'Pre-computing Rebuilds' for a detailed description of our rebuilding algorithm. 
    \item \underline{Algorithm:} Whenever a rebuild at layer $i$ finishes, it replaces all sets $B_j \gets B_{j}^{(i)}$, $A_j \gets A_{j}^{(i)}$ and $S_j \gets S_{j}^{(i)}$, and it replaces the flows $\ff_j \gets \ff_j^{(i)}$. This exactly corresponds to a $\textsc{ReInitialize}()$ operation in \Cref{alg:flow_cert_maint}. Therefore, we run the algorithm presented in section \Cref{sec:amortized_low_rec} for our new batching sequence while pre-computing the \textsc{ReInitialize} step. The layers $k - 1, \ldots, \theta$ don't need to be pre-computed in the background because they are small and can therefore be computed in a single step. 
    \item \underline{Pre-computing Rebuilds: } Only empty and half empty layers pre-compute rebuilds. Whenever layer $i \leq k - 2$ is empty or half-empty, and layer $i + 1$ becomes full, a rebuild starts pre-computing at layer $i$. We now describe how such a rebuild at layer $i$ is computed. 

    We first update the sets $A^{(i)}_i$ and $B^{(i)}_i$, by setting $B_i^{(i)} \gets B_i \cup B_{i + 1}$ and $A^{(i)}_i \gets A_i \cup S_i \cup A_{i + 1} \cup S_{i + 1}$. We then re-start an instance of \textsc{BatchPruning}() (\Cref{alg:batch_pruning}) at index $i$ of the for loop for $B_i^{(i)}, A_i^{(i)}$ and run this iteration of the loop, yielding a set $S^{(i)}_i$. Then we recall that $G_j^{(i)} := G\left[V \setminus \left(\bigcup_{l = 0}^{i - 1} S_l \cup \bigcup_{l = i}^{j - 1} S_l^{(i)}\right)\right]$, and we initialize a worst case flow certificate $\mathcal{D}_j^{(i)} \gets \textsc{WorstCaseFlow}(G^{(i)}_i, S_i^{(i)})$ from $S^{(i)}_i$ to its complement in $G_i^{(i)}$. All of these compute steps are distributed over the next $2^{k - i - 2}$ deletions, which is when $B_{i + 2}$ will become full and we start preparing the variables $X^{(i)}_{i + 1}$. 

    We next describe the steps the rebuild takes when the set $B_j$ for $j \geq i + 2$ becomes full, and the sets $B_{i + 1}, \ldots, B_{j - 1}$ are already full as well and therefore the variables $X^{(i)}_{i + 1}, \ldots,  X^{(i)}_{j - 1}$ have been initialized. The re-build first sets $B_{j - 1}^{(i)} \gets B_j$ and $A_{j - 1}^{(i)} \gets A_j \cup S_j$. Then it runs the $(j - 1)$-th iteration of imonts instance of \textsc{BatchPruning}() (\Cref{alg:batch_pruning}) yielding a set $S^{(i)}_{j - 1}$.
    
    It passes $(\bigcup_{j' \geq j - 1} B_{j'}, \bigcup_{j' \geq j - 1} B_{j'} S_{j'} \cup A_{j'})$ as a deletion batch to $\mathcal{D}_{i}^{(i)}, \ldots, \mathcal{D}_{j - 2}^{(i)}$. This ensures that the worst-case flows are supported on the graph $G[V \setminus S_0]$, because $G[V \setminus S_0]$ is a sub-graph of $G[V \setminus \bigcup_{j = 1}^\lambda S_j \cup A_j]$ because of the ensured monotonicity of pruning.   
    
    Furthermore, a new data structure $\mathcal{D}_{j - 1}^{(i)} \gets \textsc{WorstCaseFlow}(G^{(i)}_{j - 1}, S_{j - 1}^{(i)}, \deg_{G^{(0)}}, \phi)$ is initialized to compute a flow certificate routing out of $S_{j - 1}^{(i)}$ in $G^{(i)}_{j - 1}$. This compute is distributed over the next $2^{k - j - 2}$ deletions.

    We finally describe how our data-strucure handles the final layers $k - 1$ and $k$. We pre-compute the batches with index $k - 1$ and $k$ when all batches $B_j$, $j \geq i$ are full. We run the \textsc{BatchPruning}() (\Cref{alg:batch_pruning}) instance of layer $i$ to the end, compute all the flows directly, and pass the deletion sets to the data structures $\mathcal{D}_{i}^{(i)}, \ldots, \mathcal{D}_{k - 2}^{(i)}$. Finally, the algorithm additionally sets $S_{k}^{(i)} \gets \bigcup_{j \geq k} S_{j}^{(i)}$ and $A_{k}^{(i)} \gets \bigcup_{j \geq k} A_{j}^{(i)}$ to add the pruning from layers that don't contain any deletions. 

    At this point, the pre-computation finishes and the precomputed variables become the global view, i.e. $X_{j} \gets X^{(i)}_j$ for $i \geq j$. Finally, a batch flow certificate \textsc{BatchFLowCert}() is re-initialized for the new sets $S_i, \ldots, S_k$ with the flows $\ff_{j}$ precomputed by the data-structures $\mathcal{D}_{j}^{(i)}$. Since these flows and variables are already available, the $\textsc{ReInitialize}(\{S_i', \ldots, S'_k\})$ routine now runs in time $O(k)$. 

    Finally, the algorithm uses $\textsc{DeleteEdge}(e)$ of the \textsc{BatchFlowCert}() to delete $10^7 \lambda \log_2(n)/\phi$ edges incident to each set $S_j$ after every edge deletion. This ensures that the set $S_{j}$ is fully pruned when it's replaced, and therefore $G[V \setminus S_0] = G\left[V \setminus \bigcup_{i = j}^k S_j \right]$ when i.e. layer $i$ finishes its pre-compute and is about to swap in the pre-computed variables. This ensures that all the pre-computed flows are supported on the current graph. 
\end{itemize}

\paragraph{Analysis of Worst-Case Pruning Algorithm. }

We first state the main theorem of this section and our whole article. 

\begin{restatable}{theorem}{LowRecPrune}
    \label{thm:double_low_recourse}
    There is an algorithm that given a $\phi$-expander graph $G = (V, E)$ and a sequence of up to $\tilde{\Omega}(\phi |E|)$ edge deletions to $G$, and adds up to $\tilde{O}\left(\phi^{-2}\right)$ vertices to a initially empty set $A$ after each deletion such that $G[V \setminus A]$ remains a $\Omega\left(\frac{\phi}{\log_2^4 m}\right)$-expander throughout. Furthermore, every update is processed in worst-case time $\tilde{O}\left(\phi^{-2}\right)$. 
\end{restatable}

To prove \Cref{thm:double_low_recourse}, we first show the correctness of our algorithm, i.e. we show that the remaining graph remains a good expander.

\begin{lemma}[Correctness]
\label{lem:double_low_recourse_correctness}
    The algorithm $\textsc{WorstCaseLowRec}()$ maintains a set $A$ such that $G[V \setminus A]$ remains a $\Omega(\phi/\log_2^4 m)$ expander throughout given a sequence of up to $\tilde{\Omega}(\phi m)$ deletions.
\end{lemma}
\begin{proof}
    We let $A = S_0$. We show that $\left(\left(\bigcup_{i = 1}^{k} S_k\right) \setminus S_0, \ff := \sum_{i = 1}^{k} \ff_i\right)$ is a $(10k, 10k^2, 8000 k^4/\phi)$ flow certificate on $G[V \setminus S_0]$ throughout as in \Cref{sec:amortized_low_rec}. Our batching scheme ensures that batch $B_i$ contains at most $2^{k - i}$ edges. We have that $\vol_G(S_i) \leq 2(64 \log_2 n ) \lambda  \cdot 132 \cdot 2^{k - i}$, and by the description of our algorithm $2^{k - i - 1}$ operations in between the time $S_i$ got published, and it gets re-published next. Therefore, all its edges enter $S_0$, because the algorithm prunes the whole volume and an expander is always connected. Therefore, the preconditions of \Cref{thm:batch_flow_cert} are satisfied whenever a rebuild is switched, and thus the flow $\ff := \sum_{i = 1}^{k} \ff_i$ is such a flow certificate as claimed if the rebuilds are computed correctly, which we argue next. 
    
    We next show that the rebuilds are computed, i.e. that the flow $\ff_l$ routes $(10 k \cdot l + l + 1)\deg_{G^{(0)}}(v)$ flow from sources $v \in S_l \cap V$ to sinks of capacity $(10k - 1)\deg_{G^{(0)}}(v)$ for $v \in V \setminus S$ with capacity $8000k^3/\phi$. We first argue that the preconditions of \Cref{lem:init_spread}  are met. Since the number of deleted edges is less than $c \phi \vol_{G^{(0)}}/\log_2^3 m$ for some small enough constant $c > 0$ , the pruned sets are always firmly on the small side by  \Cref{clm:volume}. The size of the discarded set and further pruned sets are bounded with $2 \cdot 10^7 \log_2^3 2^{k - l}/\phi$ by \Cref{inv:item3} of \Cref{lem:main_invariant} and a geometric series in accordance with the precondition of \Cref{lem:init_spread}. Furthermore, the graph $G\left[V \setminus \bigcup_{i = 0}^k S_k\right]$ is a $\phi/10$ expander throughout since the meta algorithm maintains an expansion certificate as in \Cref{lem:exp_cert} because the final excess $2^{-1 - \log_2 \phi^{-1}}/\phi$ is less than $1$ when routing to sinks $\deg_{G}$, and this excess can be absorbed by just adding another $\deg_{G}$ sink as in \Cref{lem:exp_cert}. Here, it is crucial that $\bigcup_{j \in [\lambda]} A_j \subseteq S_0$ Therefore, \Cref{lem:init_spread} applies, and it directly follows flow certificates are initialized correctly. We remark that the additional pruning ensures that $G[V \setminus S_0] = G\left[V \setminus \bigcup_{i = 1}^k S_i \cup A_i \right]$ before a rebuild gets switched in, and since the pre-computed flows are supported on $G\left[V \setminus \bigcup_{i = 1}^k S_i \right]$ by the description of our algorithm and \Cref{lem:init_spread} they are supported on the current state of the decremental pruned graph $G$.
    
    The expansion guarantee of $G[V \setminus S_0] \setminus B$, where $B$ is the set of deleted edges, then directly follows from \Cref{lem:flow_cert_guarantees_exp} since $k = \ceil{\log_2 m}$.
\end{proof}

Then, we prove the worst case runtime guarantee and recourse. 

\begin{lemma}[Worst-Case Runtime]
    \label{lem:double_low_recourse_runtime}
    The algorithm $\textsc{WorstCaseLowRec}()$ processes each edge update in worst case update time $\tilde{O}\left(\phi^{-2}\right)$ and prunes at most $\tilde{O}
    \left(\phi^{-2}\right)$ vertices after every edge deletion given a sequence of up to $\tilde{\Omega}(\phi m)$ deletions.
\end{lemma}
\begin{proof}
    The volume of the set pruned at level $i$ is at most $2 \cdot  (64 \log_2 n)100 \lambda \cdot 2^{k - i}/\phi$ by \Cref{clm:source_reduction} and a geometric sum for $i = k$. We note that $\lambda = \ceil{\log_2 m} + \log_2 \phi^{-1} + 1$ and therefore $\lambda = \tilde{O}(1)$ because every unweighted connected graph has expansion $\phi \geq 1/m$. Running the $i$-th step in the main loop of the meta algorithm (\Cref{alg:batch_pruning}) only takes time $\tilde{O}( 2^{k - i}/\phi^2)$ by \Cref{clm:source_reduction} since its pre-conditions are again met by our batching scheme. The initialization of a worst case flow \textsc{WorstCaseFlow} at level $i$ takes at most $\tilde{O}( 2^{k - i}/\phi)$ by \Cref{clm:runtime} which can be amortized over the $2^{k - i - 2}$ deletions until the next batch gets full by executing $\tilde{O}(\phi^{-1})$ compute per step. The same holds for the amortization of each batch removal, of which there are at most $k$. Since there are up to $k$ levels pre-computing in parallel, and each of these has to execute up to $k$ worst case batch certificate builds, the total update time only increases by a factor $k^2 = O(\log^2 m)$. Finally, the algorithm from section \Cref{sec:amortized_low_rec} is run on the current state of the algorithm (without computing rebuilds), which adds at most $\tilde{O}(1/\phi^2)$ to the processing time. 
    The worst-case update time therefore follows from the description of our algorithm. 
\end{proof}

We then conclude with a proof of the main result of this article. 

\begin{proof}[Proof of \Cref{thm:double_low_recourse}]
 Follows directly from \Cref{lem:double_low_recourse_correctness} and \Cref{lem:double_low_recourse_runtime}.
\end{proof}

\newpage

\bibliographystyle{alpha}
\bibliography{refs}

\end{document}